\pdfoutput=1
\RequirePackage{ifpdf}
\ifpdf 
\documentclass[pdftex]{sigma}
\else
\documentclass{sigma}
\fi

\numberwithin{equation}{section}

\newtheorem{Theorem}{Theorem}[section]
\newtheorem{Corollary}[Theorem]{Corollary}
\newtheorem{Lemma}[Theorem]{Lemma}
\newtheorem{Proposition}[Theorem]{Proposition}
 { \theoremstyle{definition}
\newtheorem{Definition}[Theorem]{Definition}}

\def\ba#1\ea{\begin{align}#1\end{align}}
\newcommand{\prt}{\partial}

\newcommand{\dd}{\mathrm{d}}

\newcommand{\N}{{\mathbb N}}

\newcommand{\NN}{\mathbb{N}}
\newcommand{\cG}{{\cal G}}

\def\cS{\mathcal{S}}

\begin{document}

\allowdisplaybreaks

\renewcommand{\thefootnote}{$\star$}

\newcommand{\arXivNumber}{1512.02087}

\renewcommand{\PaperNumber}{056}

\FirstPageHeading

\ShortArticleName{The Multi-Orientable Random Tensor Model, a Review}

\ArticleName{The Multi-Orientable Random Tensor Model,\\ a Review\footnote{This paper is a~contribution to the Special Issue on Tensor Models, Formalism and Applications. The full collection is available at \href{http://www.emis.de/journals/SIGMA/Tensor_Models.html}{http://www.emis.de/journals/SIGMA/Tensor\_{}Models.html}}}

\Author{Adrian TANASA~$^{\dag\ddag\S}$}

\AuthorNameForHeading{A.~Tanasa}

\Address{$^\dag$~Univ. Bordeaux, LaBRI, UMR 5800, 351 cours de la Lib\'eration, 33400 Talence, France}
\EmailD{\href{mailto:adrian.tanasa@ens-lyon.org}{adrian.tanasa@ens-lyon.org}}
\URLaddressD{\url{http://www.labri.fr/perso/atanasa/}}

\Address{$^\ddag$~IUF, 1 rue Descartes, 75231 Paris Cedex 05, France}
\Address{$^\S$~H.~Hulubei National Institute for Physics and Nuclear Engineering,\\
\hphantom{$^\S$}~P.O.~Box MG-6, 077125 Magurele, Romania}

\ArticleDates{Received December 08, 2015, in f\/inal form June 10, 2016; Published online June 15, 2016}

\Abstract{After its introduction (initially within a group f\/ield theory framework) in [Tanasa~A., \textit{J.~Phys.~A: Math.
 Theor.} \textbf{45} (2012), 165401, 19~pages, arXiv:1109.0694], the multi-orientable (MO) tensor model grew over the last years into a solid alternative of the celebrated colored (and colored-like) random tensor model. In this paper we review the most important results of the study of this MO model: the implementation of the $1/N$ expansion and of the large~$N$ limit ($N$ being the size of the tensor), the combinatorial analysis of the various terms of this expansion and f\/inally, the recent implementation of a double scaling limit.}


\Classification{05C90; 60B20; 81Q30; 81T99}

\rightline{\it Dedicated to Vincent Rivasseau's 60th birthday anniversary}

\renewcommand{\thefootnote}{\arabic{footnote}}
\setcounter{footnote}{0}

\section{Introduction}

Random tensor models are a natural generalization in dimension higher than two of the celebrated (bidimensional) matrix models (see, for example, the review~\cite{DiFrancesco:1993nw}). Tensor models were f\/irst proposed in the nineties~\cite{t1, t2} but unfortunately did not draw at that moment a particular attention within the mathematical physics community.

A considerable revival of interest for these tensor models appeared with the proposition of the so-called {\it colored} tensor model \cite{color}.
Let us recall here that this proposition was actually made within the related group f\/ield theoretical (GFT) framework
(for general references on GFT, see the book \cite{oriti} or the review articles \cite{freidel, oriti2012} or \cite{BaratinOriti}).
 Shortly after, the implementation of the $1/N$ expansion for this colored model \cite{largeN, GR} largely contributed to this revival of interest. The ge\-ne\-ral term of this expansion was thoroughly analyzed, from a purely combinatorial point of view, in~\cite{GS}. This analysis allowed for the implementation of the double scaling limit mechanism for the colored tensor model. In parallel, this mechanism has been implemented for a closely related model, in \cite{DGR}, using this time purely quantum f\/ield theoretical (QFT) techniques (namely, the intermediate f\/ield method).

Moreover, after adding an appropriate Laplacian operator in the action of this type of models, several renormalizability studies have been made. The f\/irst perturbative renormalizable tensor model was the Ben Geloun--Rivasseau tensor model \cite{BGR}. The combinatorics of the renormalizability of this model has been expressed in a Hopf algebraic setting in \cite{RT-CK}.
Several other models have been proved renormalizable, models def\/ined this time in closer relation to the initial GFT framework~-- see the thesis \cite{teza-Sylvain} and references within. Various $\beta$-function computations have also been made, see \cite{ART, Sylvain-AIHPD} and references within. Moreover, in~\cite{ART}, the combinatorics of the Dyson--Schwinger equation of a particular such tensor model has been analyzed from a purely algebraic point of view.

For the sake of completeness, let us also mention that some relations between tensor models and matrix models have been investigated in \cite{BC2-AIHPD2}. Moreover, relations between tensor models and meanders have been studied in \cite{BC1-AIHPD2}. Finally, some relations between the counting of tensor model observables and branched covers of the two-sphere have been proved in \cite{BGR-AIHPD}.

For various reviews on tensor models, we point the interested reader to the reviews \cite{GR-Sigma, rev-riv} and \cite{praa}.

Getting back to colored (and colored-like) tensor models, we point out that they have a major drawback: an important class of Feynman tensor graph are discarded, by the very def\/inition of these models. This drawback is softened when working with the multi-orientable (MO) tensor model, where it was proved that a larger class of Feynman graph is kept.

This is, in our opinion, the main motivation for the study of the MO model. One can remark that this motivation is not of geometric or topological nature, but instead it has a strong QFT f\/lavor. Indeed, the QFT philosophy is to carefully look over all classes of graphs, and not to chose just some classes of Feynman graphs which are easier to study. For example, if in Moyal QFT one discards the non-planar sector, the celebrated UV/IR mixing phenomena is lost.

With this QFT optics in mind, we do not consider that the MO model is the most general tensor model one should study, while keeping interesting properties such as the large $N$ expansion or the double scaling limit (see also the perspectives described in the last section of this review). Nevertheless, we consider that, with respect to the colored and colored-like models, the MO model is an interesting step ahead in this direction.

The MO model was def\/ined, again within a GFT setting, in \cite{original} (see also \cite{praa-mo} for a short review). The $1/N$ expansion and the large $N$ limit have been implemented in \cite{DRT}. The sub-dominant term of this expansion has been studied in \cite{RT}. A thorough combinatorial analysis of the general term of this expansion, has been done in \cite{FT}. This has then led to the recent implementation of the double scaling mechanism \cite{GTY}. All these results are presented within the following four sections. The last section of this review is dedicated to some concluding remarks and to a list of perspectives for future work.

\section{Def\/inition of the model}

In this section, we give the def\/inition of the MO random tensor model. This follows \cite{DRT}.

Let $\phi_{ijk}$, $i,j,k = 1, \ldots, N$, be the components of a three index complex tensor $\phi$.
The action of the MO tensor model writes
\begin{gather}\label{eq:S}
 S[\phi] = \sum_{ijk} \phi_{ijk} \bar{\phi}_{ijk} - \frac{\lambda}{2} \sum_{\genfrac{}{}{0pt}{}{ijk}{lmn}}\phi_{ijk} \bar{\phi}_{mlk} \phi_{mjn} \bar{\phi}_{iln}.
\end{gather}
Note that the $1/2$ factor multiplying the coupling constant $\lambda$ takes into consideration the symmetries of the vertex.

The partition function of the model
\begin{gather*}
 Z = \int \mathcal{D}[\phi] e^{ -S[\phi] } , \qquad \mathcal{D}[\phi] = \prod_{ijk} \frac{{\rm d}\phi_{ijk} \, {\rm d}\bar{\phi}_{ijk}}{2\pi\imath} ,
\end{gather*}
writes in perturbation theory as a sum over Feynman MO tensor graphs.
From a combinatorial point of view, the partition function is a generating function of graphs to whom one associates a certain weight (which is actually the Feynman amplitude below).

The vertices of the MO graphs are four valent, the edges are oriented from $\phi$ to $\bar \phi$ and
the orientations alternate around a vertex. The edges can be represented as three parallel strands (one
for each index of the tensor) and the vertices as the intersection of four half edges such that
every pair of half edges shares a~strand. This orientation can also be seen in the following way. One labels the four corners of the vertex with alternating signs `$+$' and `$-$', see Fig.~\ref{graf}. An edge of the graph then connects a `$+$' to a `$-$' sign. In Fig.~\ref{graf}, an example of a vacuum MO tensor graph is given.

\begin{figure}[htb]
\centering
\includegraphics[scale=0.6]{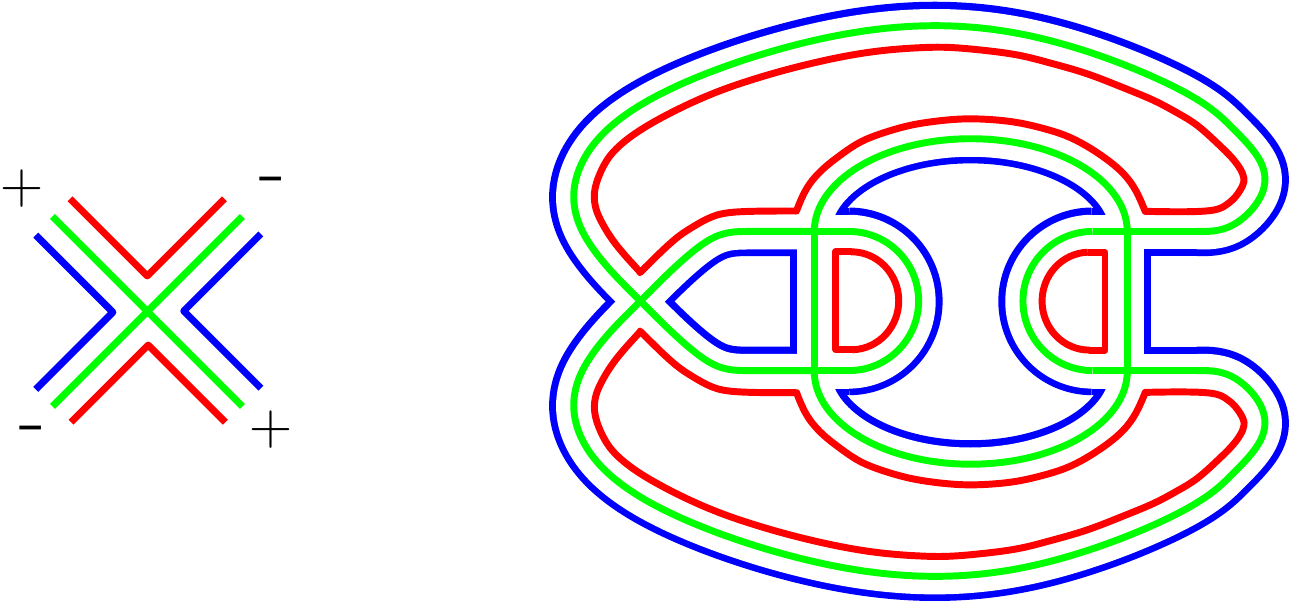}
\caption{The MO vertex and an MO vacuum Feynman graph.\label{graf}}
\end{figure}

The name {\it multi-orientable} was given in \cite{original}, because this type of model enjoins two ``kinds of orientability'':
\begin{enumerate}\itemsep=0pt
\item Orientability at the level of the propagator: no twists are allowed between the various strands. In the original GFT literature, the name ``orientable'' was given for models with this property (see for example~\cite{fgo} and GFT references within).
\item Orientability at the level of the vertex. This type of vertex is called ``orientable'' in the Moyal non-commutative QFT literature (see, for example, the review \cite{rev-riv-ncqft} and references within). Within this framework, the vertex is pictured as in Fig.~\ref{vertexncqft}.
\end{enumerate}

\begin{figure}[htb]
\centering
\includegraphics[scale=0.2]{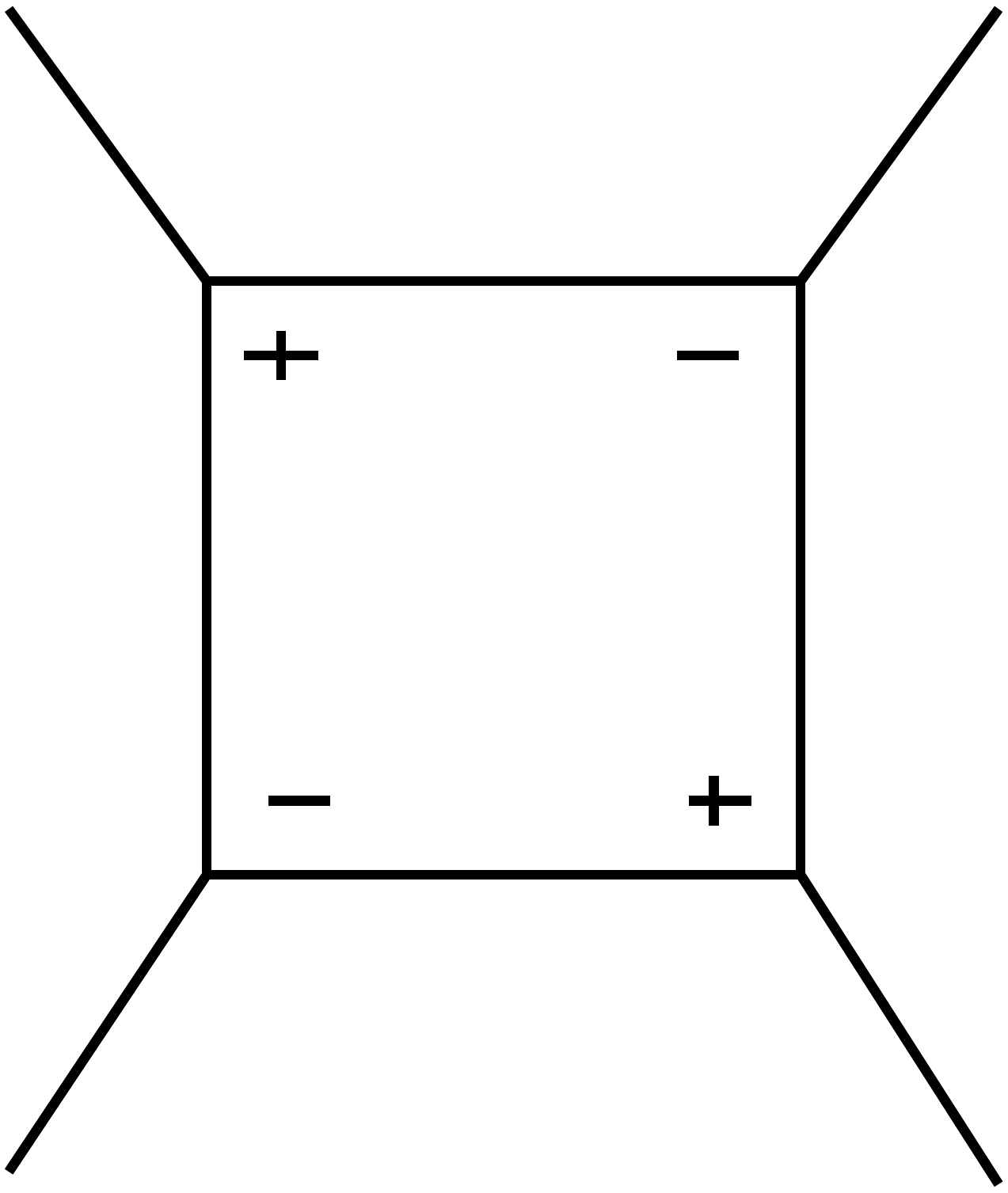}
\caption{A dif\/ferent representation (without strands) of the MO vertex.\label{vertexncqft}}
\end{figure}

One can then prove that the MO action \eqref{eq:S} leads to a set of Feynman graphs which is strictly larger than the set of Feynman graphs of the colored model (for more details, see \cite{original}) or colored-like models.
Let us recall here that colored models allow only for Feynman graphs where each edge has a color, usually labeled in 3D, from $0$ to $3$. At each vertex, one has a half-edge of each of these four colors~-- the graph is bipartite. Moreover, a face of the graph is formed by edges of exactly two colors.
The MO model allows for all these Feynman graphs. Moreover, its perturbative expansion allows as well for the following supplementary classes of Feynman graphs:
\begin{enumerate}\itemsep=0pt
\item MO graphs which are not edge-colorable (with at most four colors) but are not allowed by the colored model action,
\item MO graphs which are edge-colorable (with at most four colors) but are not allowed by the colored model action,
\item MO graphs which are not bipartite.
\end{enumerate}
Note that edge-colorability implies bipartiteness~-- this is a standard result of graph theory.
Examples of such MO tensor graphs are given in Fig.~\ref{planartadtwistsun}. On the left one has a double tadpole graph, which is an MO graph which is not edge-colorable. On the right one has a graph which is edge
colorable, using four colors, $0,\ldots, 3$, but does not occur in colorable models.

\begin{figure}
\centering
\includegraphics[scale=0.99]{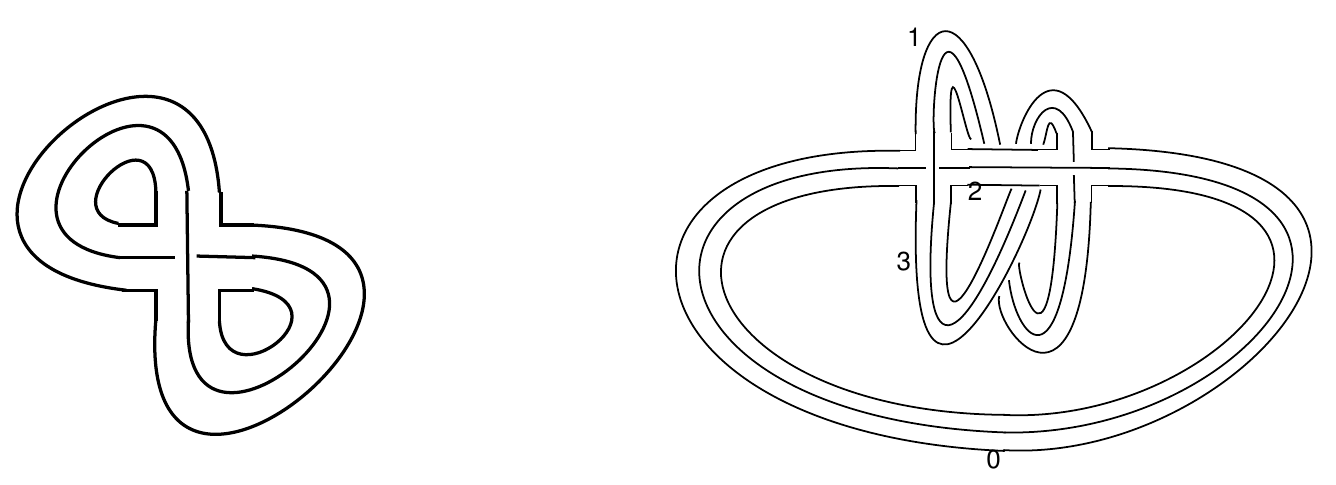}
\caption{Examples of MO tensor graphs which do not occur in the colorable framework.\label{planartadtwistsun}}
\end{figure}

The strands are divided into three classes:
the ones in the middle of the edges, called the~$S$ (for straight) strands,
the ones on the right (with respect to the $\phi\to \bar \phi$ orientation) of the edges, called the $R$ (for right) strands
and the ones on the left of the edges, called the $L$ (for left) strands. At a vertex the~$S$ ($R$ or $L$) strands
only connect to~$S$ ($R$ or~$L$) strands.

Note that MO tensor graphs are in one-to-one correspondence with maps. This is obtained by getting rid of the $S$ strand. Thus, one can use strand-less graphical representations of MO tensor graphs (such as the one given in Fig.~\ref{vertexncqft} of in Section~\ref{sec:fusy}).

To any MO graph one can associate three canonical ribbon graphs, called the {\it jackets},
obtained by erasing throughout the graph all the strands in the same class ($S$, $L$ or $R$).
For the graph of Fig.~\ref{graf}, this leads to the three jackets represented in Fig.~\ref{jachete}.
\begin{figure}[htb]
\centering
\includegraphics[scale=.5]{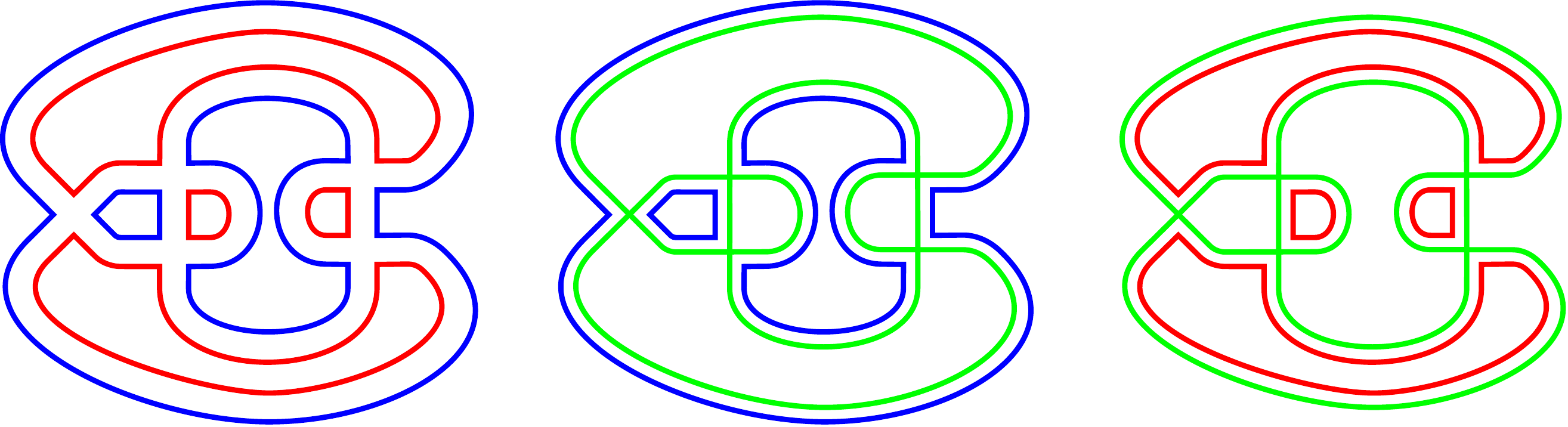}
\caption{The three jackets associated to the graph of Fig.~\ref{graf}.\label{jachete}}
\end{figure}

The fact that this algorithm always leads, for an MO graph, to a ribbon graph was proved in Proposition~4.1 of~\cite{DRT}.
Let us also emphasize that this does not hold for general non-MO graph
(see Fig.~\ref{algotadface} for such a counterexample).

\begin{figure}[htb]
\centering
 \includegraphics[scale=0.18]{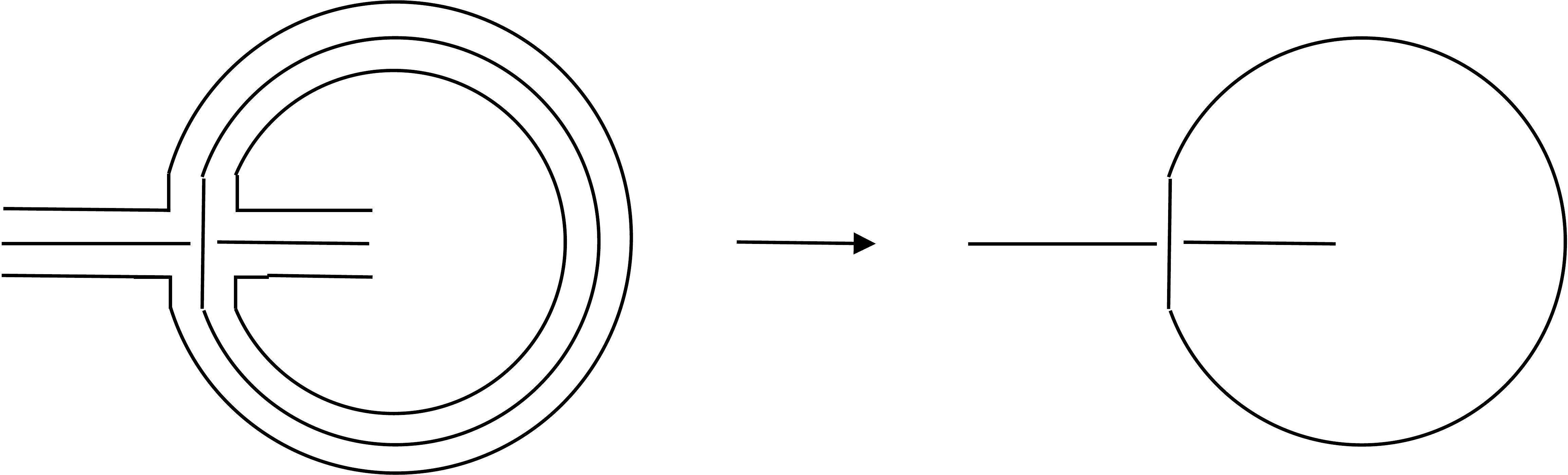}
 \caption{Deletion a pair of opposite corner strands in a non-MO graph.\label{algotadface}}
\end{figure}

We further call the {\it canonical genus} of an MO graph the genus of the jacket obtained by erasing the straight faces (we also call this particular jacket the {\it canonical jacket}). Thus, a~planar MO tensor graph is an MO graph with vanishing canonical genus.

Let us emphasize here that the jackets of the MO tensor graph can be non-orientable. This is a fundamental dif\/ference between the MO model and the colored-like models.

\begin{Definition}The {\it degree} $\delta(\cG)$ of a connected MO graph $\cG$ is the half sum of the non
orientable genera of its three jackets
\begin{gather*}
 \delta(\cG)=\frac{1}{2}\left( k_{{\cal J}_1} + k_{{\cal J}_2} + k_{{\cal J}_3} \right),
\end{gather*}
where the parameters $ k_{{\cal J}_i}$ ($i=1,2,3$) are the non-orientable genera of the three jackets of the graph.
\end{Definition}

As a direct consequence of this def\/inition, one concludes that the degree is a positive half integer.

\section[The $1/N$ expansion and the large $N$ limit]{The $\boldsymbol{1/N}$ expansion and the large $\boldsymbol{N}$ limit}

In this section we implement the $1/N$ expansion and we study the large $N$ limit of the model. In the f\/irst two subsections, we analyze the expression of the Feynman amplitudes and the leading order of the $1/N$ expansion. This follows the original article~\cite{DRT}.
In the third subsection, we analyze the next-to-leading order of the $1/N$ expansion.
This follows the original article~\cite{RT}.
In the last subsection, we study the leading and next-to-leading order series of the model. This follows again the original article \cite{RT}.

\subsection[Feynman amplitudes; the $1/N$ expansion]{Feynman amplitudes; the $\boldsymbol{1/N}$ expansion}\label{expansion}

Let us now organize the free energy series according to powers of $N$. Since each face corresponds to a closed cycle of Kronecker $\delta$ functions, each face contributes with a factor $N$. The Feynman amplitude of an MO graph then writes
\begin{gather*} 
 A= \lambda^{V} (k_N)^{-V} N^{F},
\end{gather*}
where $V$ is the number of vertices of the graph, $F$ is the number of faces of the graph and $k_N$ is a rescaling constant.
 We choose this rescaling $k_N$ to get the same divergence degree for the leading graphs at any order.
We f\/irst count the faces of a graph using the previously def\/ined jackets~$\mathcal{J}$.
Using the Euler characteristic formula, one has
\begin{gather} \label{facejacket}
f_{\mathcal{J}}= e_{\mathcal{J}}-v_{\mathcal{J}}-k_{\mathcal{J}}+2,
\end{gather}
where $k_{\mathcal{J}}$ is the non-orientable genus of the jacket ${\mathcal{J}}$ (see above).
Since each jacket is a connected vacuum ribbon graph, one has: $e_{\mathcal{J}} = 2v_{\mathcal{J}}$.
Let us recall here that the numbers of vertices (resp.~edges) of a~jacket
$\mathcal{J}$ of a graph is the same than the numbers of vertices (resp.~edges) of the graph.
Since each graph has three jackets and each face of a graph occurs in two jackets,
summing~\eqref{facejacket} over all the jackets of a graph leads to
\begin{gather} \label{magica}
F=\frac{3}{2}V+3-\sum_{\mathcal{J}} \frac{k_{\mathcal{J}}}{2} =\frac{3}{2}V+3-\delta.
\end{gather}
The amplitude rewrites as
\begin{gather*} 
 A= \lambda^{V} (k_N)^{-V} N^{\frac{3}{2}V+3-\delta}.
\end{gather*}
To to get the same divergence degree for the leading graphs at any order, we choose the scaling constant~$k_N$
to be equal to $N^{\frac{3}{2}}$. The amplitude f\/inally writes as
\begin{gather} \label{3rdamplitude}
A = \lambda^{V}N^{3-\delta}.
\end{gather}
Note that, as in the case of random matrices, random tensor models (such as the MO model described in this paper), have a purely combinatorial expression. This comes from the fact that the propagator in the action~\eqref{eq:S} is, from a QFT point of view, trivial: no Laplacian operator is present in the quadratic part. As already mentioned in the Introduction, adding an appropriate Laplacian operator to the action allows for more involved expression of the associated Feynman amplitudes and for renormalizability studies (see again the thesis~\cite{teza-Sylvain}).

Using now the expression \eqref{3rdamplitude} for the Feynman amplitude we can rewrite the free energy $E$ as a formal series in $1/N$:
\begin{gather*} 
E=\sum_{\delta\in \N/2}C^{[\delta]}(\lambda)N^{3-\delta},
\end{gather*}
where
\begin{gather*}
C^{[\delta]}(\lambda)=\sum_{\mathcal{G}, \delta(\mathcal{G})=\delta} \frac{1}{s(\mathcal{G})}\lambda^{v_{\mathcal{G}}}.
\end{gather*}

\subsection[The large $N$ limit -- the leading order (melonic graphs)]{The large $\boldsymbol{N}$ limit -- the leading order (melonic graphs)}

Having implemented the $1/N$ expansion, one easily sees that the leading order in the large~$N$ limit (i.e., the limit $N\to\infty$) is given by the tensor graphs satisfying the condition
\begin{gather}\label{conditia}
\delta=0.
\end{gather}
In this subsection we identify these tensor graphs and we show that they are the so-called {\it melonic graphs}. These graphs are obtained by insertions of a fundamental two-point melonic subgraphs in the so-called {\it elementary melon} (see Fig.~\ref{insertion}).
\begin{figure}[htb]
\centering
\includegraphics[scale=0.25]{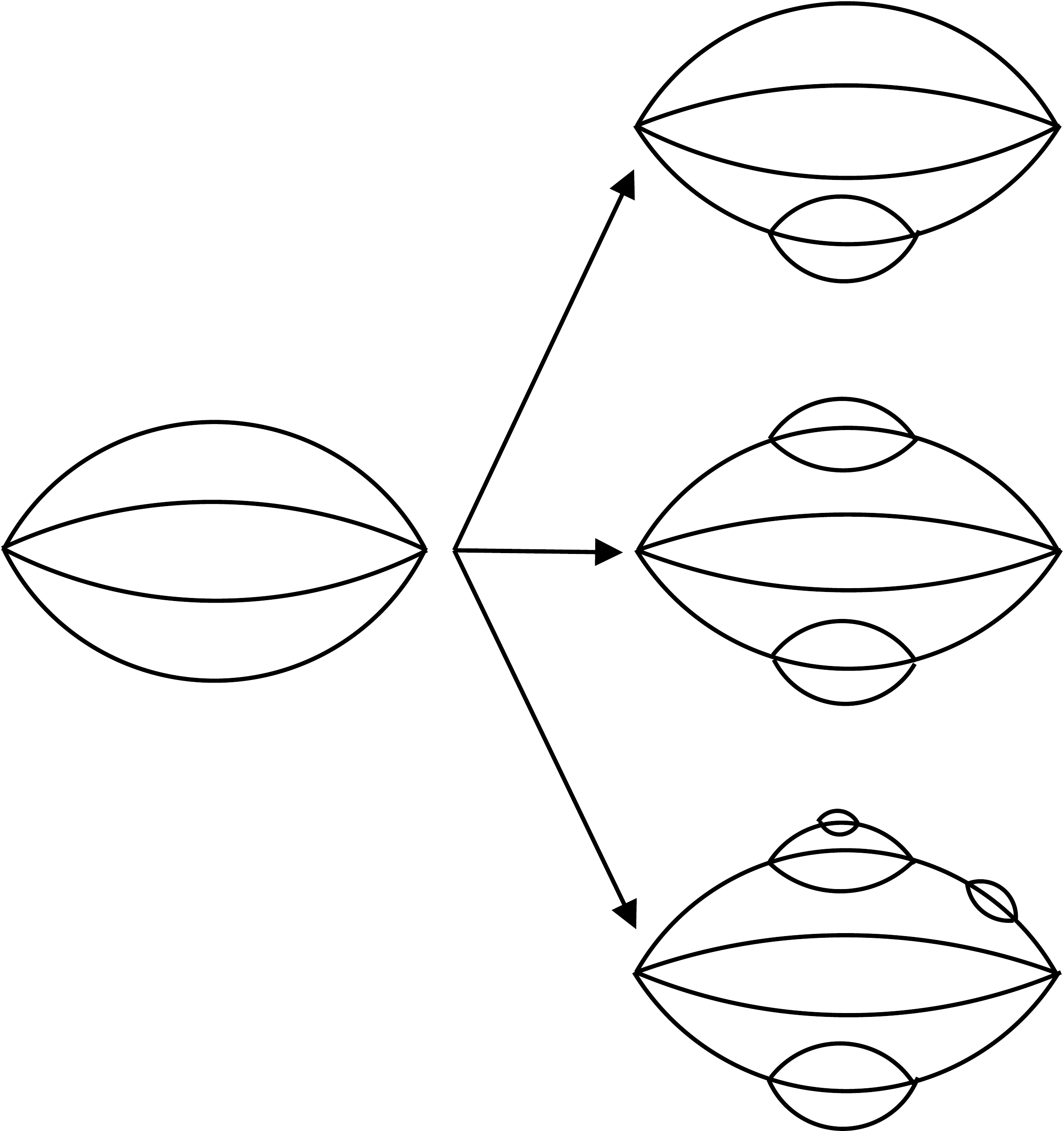}
\caption{Insertions of melonic subgraphs in the elementary melon graph.\label{insertion}}
\end{figure}
Another example of a graph obtained in this way is given in Fig.~\ref{exemplu} (example where we have used this time the stranded representation of graphs).

\begin{figure}[htb]
\centering
\includegraphics[scale=0.15]{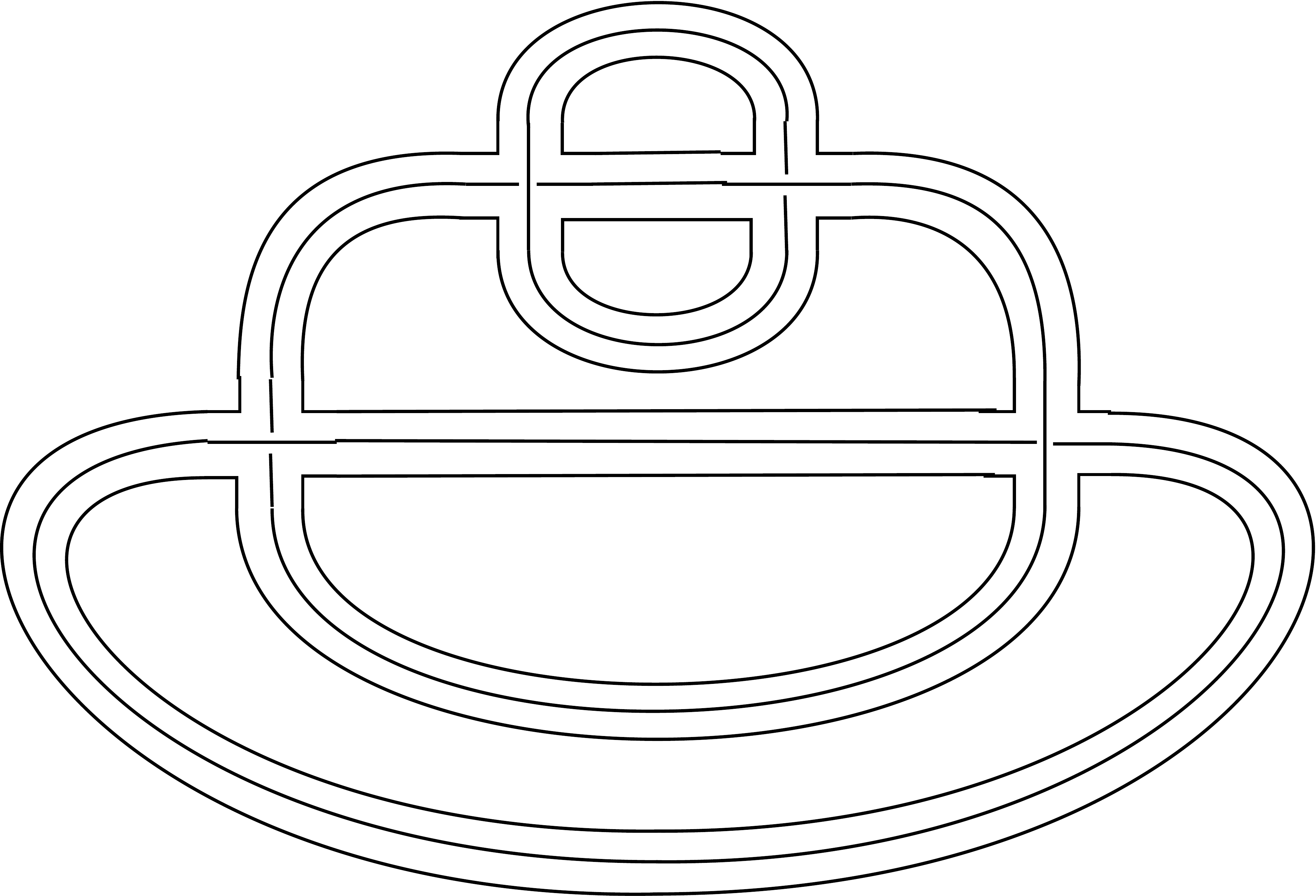}
\caption{An example of a melonic tensor graph (using the stranded representation).}\label{exemplu}
\end{figure}

Note that one can prove the following crucial fact: melonic insertions preserve the degree. This can be checked by carefully counting the number of faces and vertices and by using the expression~\eqref{magica} for the degree.

One can f\/irst prove that condition~\eqref{conditia} cannot be satisf\/ied if the graph is non-bipartite:
\begin{Proposition}\label{nbp}
A non-bipartite MO graph has at least one non-orientable jacket and thus its degree satisf\/ies the inequality: $\delta \ge \frac{1}{2}$.
\end{Proposition}

In order to identify the dominant graphs among the bipartite MO graphs, one can prove the following results:
\begin{Proposition} If $\mathcal{G}$ is a bipartite, vacuum graph of degree zero, then $\mathcal{G}$ has a face with two vertices.
\end{Proposition}

\begin{Proposition} If $\mathcal{G}$ is null degree bipartite vacuum graph, then it contains a three-edge colored subgraph with exactly two vertices.
\end{Proposition}
Using these two propositions, one can prove that the only graphs satisfying condition \eqref{conditia} are the melonic graphs (see \cite{DRT} for details).
This leads to:

\begin{Theorem}
The leading order graphs of the $1/N$ expansion of the MO model \eqref{eq:S} are the melonic ones.
\end{Theorem}

Let us end this subsection by stating two important properties of melonic graphs:
\begin{enumerate}\itemsep=0pt
\item Melonic graphs maximize the number of faces for a given number of vertices.
\item Melonic graphs correspond to a particular class of triangulations of the three-dimensional sphere $S^3$.
\end{enumerate}
Moreover, from a probabilistic point of view, these graphs were proven to correspond to branched polymers (that is, they possess Hausdorf\/f dimension two and spectral dimension $4/3$) \cite{GR-AIHP}.

\subsection[The large $N$ limit -- the next-to-leading order]{The large $\boldsymbol{N}$ limit -- the next-to-leading order}

Using again the $1/N$ expansion of Subsection~\ref{expansion}, one easily sees that the next-to-leading order in the large $N$ limit is given by the tensor graphs satisfying the condition
\begin{gather*}
\delta=\frac 12.
\end{gather*}
In this subsection we identify these tensor graphs and we prove that they are the so-called {\it infinity graphs}. These graphs are obtained by melonic insertions in the double tadpole graph of Fig.~\ref{planartadtwistsun}.

One can then prove (see \cite{RT} or \cite{FT} for details) the following result:

\begin{Theorem}\label{thm:main}
The next-to-leading order graphs of the $1/N$ expansion of the MO model~\eqref{eq:S} are the infinity ones.
\end{Theorem}

\subsection{Leading and next-to-leading order series}

Following \cite{RT}, we compute in this subsection the radiuses of convergence and the susceptibility exponents of the leading and next-to-leading order series of the MO tensor model.

\subsubsection{Leading order series}

The analysis of the leading order (LO) series of the MO model is identical to the one of the colored model (since one is interested to the same class of graphs, the melonic graphs, and these graphs have the same Feynman amplitude in the two cases).

We denote by $\lambda_{\rm c,LO}$ the LO critical value of the coupling constant $\lambda$. From a mathematical point, $\lambda_{\rm c,LO}$ is the radius of convergence of the LO series in $\lambda$.

The leading order free energy is obtained by summing over the amplitudes of melonic vacuum graphs. We are interested in the asymptotic behavior of this LO series around $\lambda_{\rm c,LO}$

The LO connected two-point function writes
\begin{gather*}
	G_{\rm LO}(\lambda) \sim \operatorname{const} + \left( 1 - \frac{\lambda^2}{\lambda_{\rm c,LO}^2}\right)^{\frac{1}{2}}
\end{gather*}
around the critical value $\lambda_{\rm c,LO}$ of the coupling constant in the leading order. This is related to the behavior of the LO free energy, for which we obtain
\begin{gather*}
	E_{\rm LO}(\lambda) \sim \left( 1 - \frac{\lambda^2}{\lambda_{\rm c,LO}^2}\right)^{2-\gamma_{\rm LO}},
\end{gather*}
with the susceptibility exponent (or the critical exponent), being
\begin{gather*}
\gamma_{\rm LO}=\frac{1}{2}.
\end{gather*}

\subsubsection{Next to leading order series}

In order to study the behavior of the NLO series, we need to study the connected NLO two-point function. The graphs contributing to the connected NLO two-point function can be obtained from the NLO vacuum graphs by cutting any one of the internal lines of an NLO vacuum graph. Thus, we can in a straightforward
manner import the classif\/ication of NLO vacuum graphs obtained in the previous section to the case of connected NLO two-point graphs. We express the NLO
two-point function in terms of the LO two-point function through algebraic identities relating the LO and NLO two-point functions. More specif\/ically, any
two-point function is of the form: bare propagator multiplied by a specif\/ic function. We will denote this function associated to the connected LO two-point function as $G_{\rm LO}$, the function associated to the connected NLO two-point function as $G_{\rm NLO}$, and the function associated to the one-particle-irreducible (1PI) NLO two-point function as $\Sigma_{\rm NLO}$.

The f\/irst identity for the two-point functions illustrated in Fig.~\ref{fig:gnlo}, states that the connected NLO two-point function $G_{\rm NLO}$ (on the l.h.s.\ of the f\/igure) is obtained by gluing connected LO two-point functions $G_{\rm LO}$ on both sides of the 1PI NLO two-point function $\Sigma_{\rm NLO}$ (on the r.h.s.\ of the f\/igure).

This writes
\begin{gather*}
	G_{\rm NLO} = G_{\rm LO}^2 \Sigma_{\rm NLO}.
\end{gather*}

\begin{figure}[htb]\centering
\includegraphics[scale=0.5]{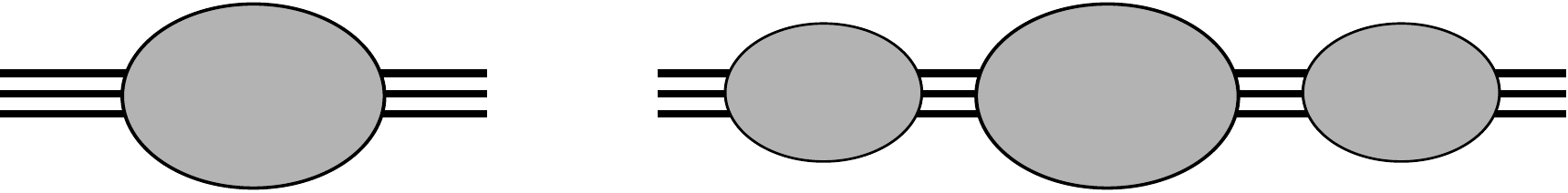}
\caption{The connected NLO 2-point function (l.h.s.\ of the f\/igure) is obtained by gluing connected LO 2-point functions on both sides of the 1PI NLO 2-point function (r.h.s.\ of the f\/igure).\label{fig:gnlo}}
\end{figure}

The second identity comes from the two dif\/ferent combinatorial ways of obtaining 1PI NLO two-point graphs from connected LO and NLO two-point graphs, and writes
\begin{gather}\label{ecuatie}
	\Sigma_{\rm NLO} = \lambda G_{\rm LO} + 3\lambda^2 G_{\rm LO}^2 G_{\rm NLO},
\end{gather}
where the combinatorial factor three arises from the three dif\/ferent internal lines of the elementary melon, on which the NLO two-point function can be inserted. The two terms above exhaust all contributions to the 1PI NLO two-point function, since any such graph contains only one tadpole
due to Theorem~\ref{thm:main}, and thus either all melonic subgraphs are inside the tadpole
(the f\/irst term) or the tadpole is a subgraph of a melonic graph (the second term).

This is illustrated in Fig.~\ref{fig:Snlo}, where the 1PI NLO two-point function $\Sigma_{\rm NLO}$ is represented on the l.h.s.\ of the f\/igure. The remaining two drawings represent the two terms on the r.h.s.\ of~\eqref{ecuatie}.

\begin{figure}[htb]\centering
\includegraphics[scale=0.5]{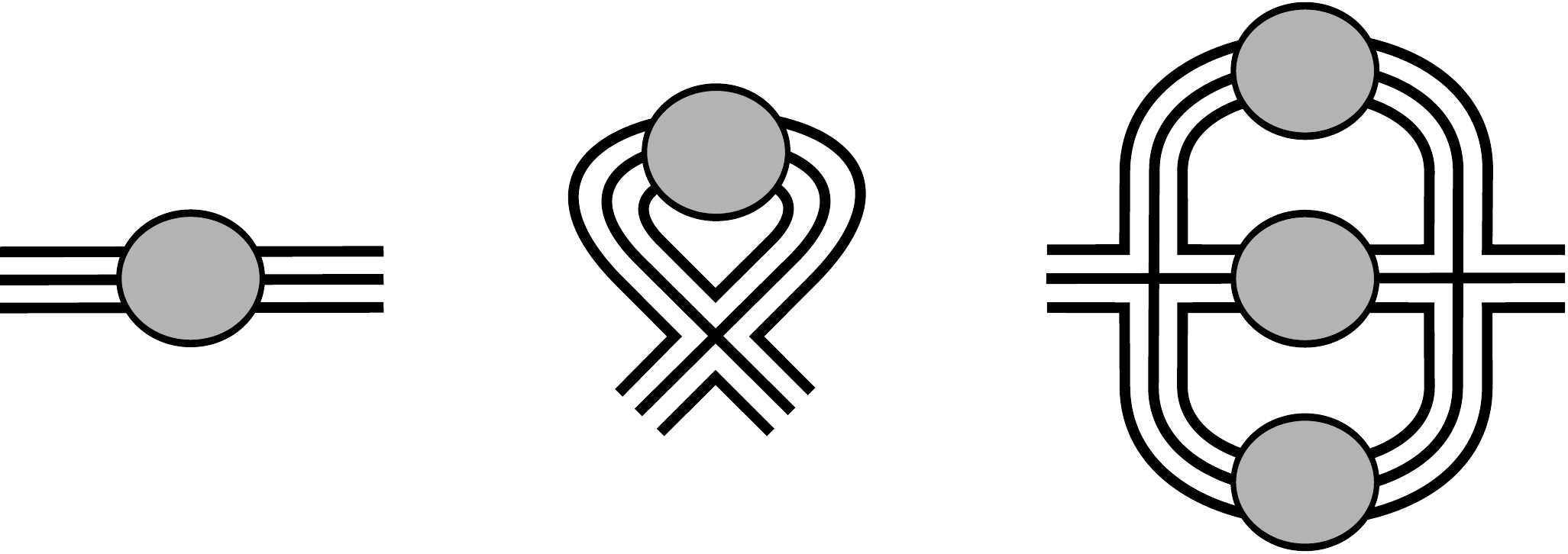}
\caption{Illustration of~\eqref{ecuatie}. The 1PI NLO two-point function $\Sigma_{\rm NLO}$ is represented on the l.h.s.\ of the f\/igure; the remaining two drawings represent the two terms on the r.h.s.\ of~\eqref{ecuatie}.\label{fig:Snlo} }
\end{figure}

Now, putting these two identities together, we obtain f\/irst the identity
\begin{gather*}
	G_{\rm LO}^{-2} G_{\rm NLO} = \lambda G_{\rm LO} + 3\lambda^2 G_{\rm LO}^2 G_{\rm NLO} ,
\end{gather*}
from which we can solve for $G_{\rm NLO}$ in terms of the connected LO two-point function
\begin{gather*}
	G_{\rm NLO} = \frac{\lambda G_{\rm LO}^3}{1-3\lambda^2 G_{\rm LO}^4} .
\end{gather*}
On the other hand, dif\/ferentiating the LO two-point function relation $G_{\rm LO} = 1 + \lambda^2 G_{\rm LO}^4$ we get
\begin{gather*}
	\frac{\prt}{\prt\lambda^2} G_{\rm LO} = \frac{G_{\rm LO}^4}{1 - 4\lambda^2 G_{\rm LO}^3} = \frac{G_{\rm LO}^5}{1 - 3\lambda^2 G_{\rm LO}^4},
\end{gather*}
where for the last equality we used the LO two-point function identity again. Thus, we get the expression
\begin{gather*}
	G_{\rm NLO} = \frac{\lambda}{G_{\rm LO}^2} \frac{\prt}{\prt\lambda^2} G_{\rm LO} ,
\end{gather*}
which implies, together with $G_{\rm LO} \sim \operatorname{const} + (1 - (\lambda^2/\lambda_c^2))^{1/2}$,
\begin{gather*}
	G_{\rm NLO} \sim \left(1 - \frac{\lambda^2}{\lambda_c^2}\right)^{-1/2}.
\end{gather*}

Finally, we use the following Dyson--Schwinger equation
\begin{gather*}
	0 = \int \dd\bar{\phi}\, \dd\phi\, \frac{\delta}{\delta\phi_{ijk}} \big( \phi_{i'j'k'} e^{-S[\phi,\hat{\phi}]} \big).
\end{gather*}
We thus obtain the relation
\begin{gather}\label{eq:SDeq}
	G_{\rm NLO} = 1 - 4\lambda^2 \frac{\prt}{\prt\lambda^2} E_{\rm NLO},
\end{gather}
relating the connected two-point function $G_{\rm NLO}$ to the free energy $E_{\rm NLO}$. Accordingly, we have $E_{\rm NLO} \sim (1 - (\lambda/\lambda_c)^2)^{1/2}$ from (\ref{eq:SDeq}), and thus f\/ind the same critical value of the coupling constant (i.e., the radius of convergence) for the NLO series (as series in the coupling constant $\lambda$) as for the LO series.
Nevertheless, one has a distinct value for the NLO susceptibility exponent
\begin{gather*}
\gamma_{\rm NLO}=\frac 32.
\end{gather*}

\section{Some combinatorial developments}\label{sec:fusy}

In this section we perform a thorough combinatorial analysis of the general term of the $1/N$ expansion of the previous section. This follows the original article~\cite{FT}.

Note that the `$+$' and `$-$' signs canonically induce an orientation of the edges. One can thus represent MO tensor graphs as a particular class of {\it oriented four-regular maps}~-- see Fig.~\ref{map} for the representation of the vertex of the MO tensor graph of Fig.~\ref{graf} in such a way.
Moreover, we work in this section which rooted maps, i.e., a connected map with a marked edge. This allows to better handle symmetries issues when counting such combinatorial objects. In physics language, marking an edge transforms a vacuum graph into a two-point Feynman graph.
One can further see this root as a fake vertex of valence two placed on some root edge.
\begin{figure}[htb]
\centering
\includegraphics[scale=0.5]{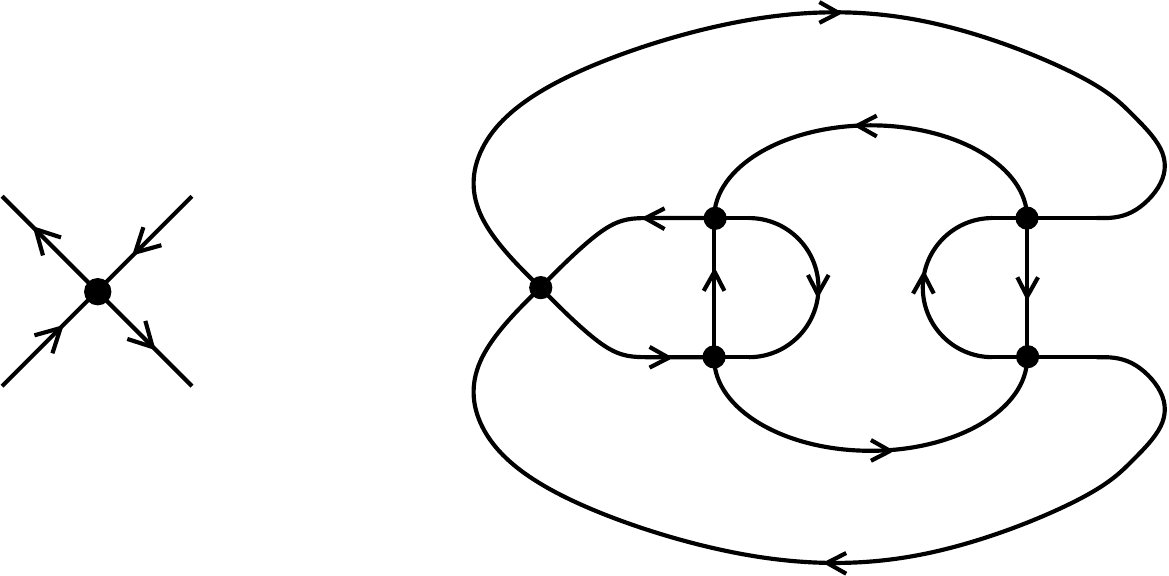}
\caption{Representation of MO tensor graphs as oriented four-regular maps.}\label{map}
\end{figure}

We then def\/ine the \textit{cycle graph} as an oriented self-loop carrying no vertex. The cycle graph is connected, has $V=0$, $F=3$ (one face in each
type), hence has degree~$0$. In its rooted version, the (rooted) \emph{cycle-graph} is made of an oriented tadpole incident to the root-vertex.

In this context, def\/ine the \emph{removal} of a melon as the operation represented in Fig.~\ref{reduce_melon} (where possibly $u=v$, and possibly $u$ or $v$ might be the root if the MO-graph is rooted).
\begin{figure}\centering
\includegraphics[width=8cm]{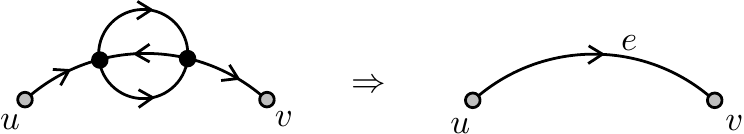}
\caption{Melon removal.\label{reduce_melon}}
\end{figure}
The reverse operation (where~$e$ is allowed to be a tadpole, and is allowed to be incident to the root-vertex if the MO-graph is rooted) is called the \emph{insertion} of a melon at an edge.

Analogously, one can def\/ine a tadpole (or loop, in graph theoretical language) removal~-- see Fig.~\ref{fig:remove_loop}.
\begin{figure}\centering
\includegraphics[width=4cm]{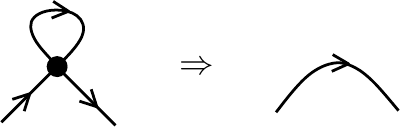}
\caption{Removing a tadpole in an MO-graph.}\label{fig:remove_loop}
\end{figure}

One then has
\begin{Lemma}
\label{lem:remove_loop}
Let ${\mathcal{G}}$ be an MO-graph of degree $\delta$,
let $\ell$ be a tadpole of ${\mathcal{G}}$, and let ${\mathcal{G}}'$ be the MO-graph obtained
from ${\mathcal{G}}$ by erasing the tadpole and its incident vertex, as shown in Fig.~{\rm \ref{fig:remove_loop}}. Then~${\mathcal{G}}'$ has degree $\delta-1/2$. Hence~${\mathcal{G}}$ has at most $2\delta$ tadpoles.
\end{Lemma}
\begin{proof}
The resulting graph ${\mathcal{G}}'$ has the same number of connected components as ${\mathcal{G}}$, has one vertex less, and has one face less (which has length one).
\end{proof}

A melonic graph (see above) can then be seen as an MO graph which can be reduced to the rooted cycle graph by successive removal of melons.

\begin{figure}\centering
\includegraphics[width=12cm]{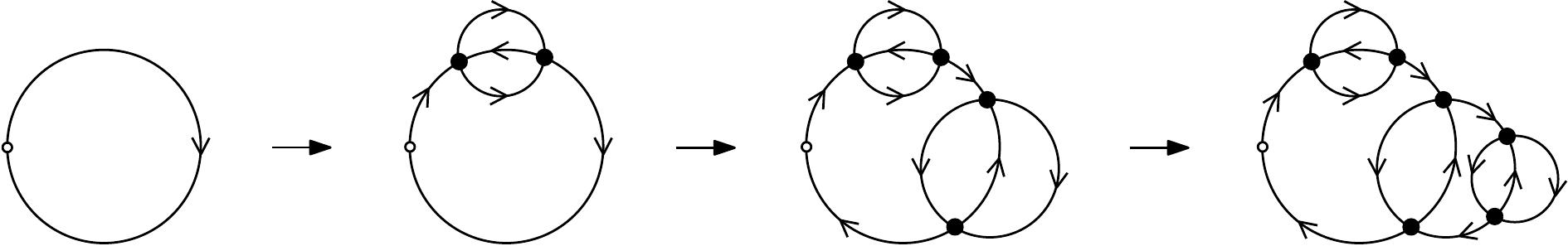}
\caption{A rooted melonic graph can be built from the cycle-graph, by melonic insertions.}\label{fig:melonic}
\end{figure}

Getting back to melonic insertions (see above), one needs to stress that {\it there is an infinite number of melon free graphs of a given degree}. Nevertheless, as we will explain the sequel, some particular types of subgraphs can be repeated without increasing the degree!

\subsection{Dipoles, chains, schemes and all that}

Let us now give the following def\/inition:

\begin{Definition}
A {\it dipole} is a subgraph of a (possibly rooted) MO-graph formed by a couple of vertices connected by two parallel edges containing a face of length two incident to these distinct vertices, and not passing by the root if the graph is rooted.
\end{Definition}
Accordingly, one has three distinct types of dipoles: L, R, or S, see Fig.~\ref{fig:dipoles} (left part).
As the f\/igure shows, each dipole has two exterior half-edges on one side and two exterior
half-edges on the other side.
Notice also that a double edge does not necessarily delimit a dipole, as shown in Fig.~\ref{fig:dipoles} (right part).

\begin{figure}\centering
\includegraphics[width=12cm]{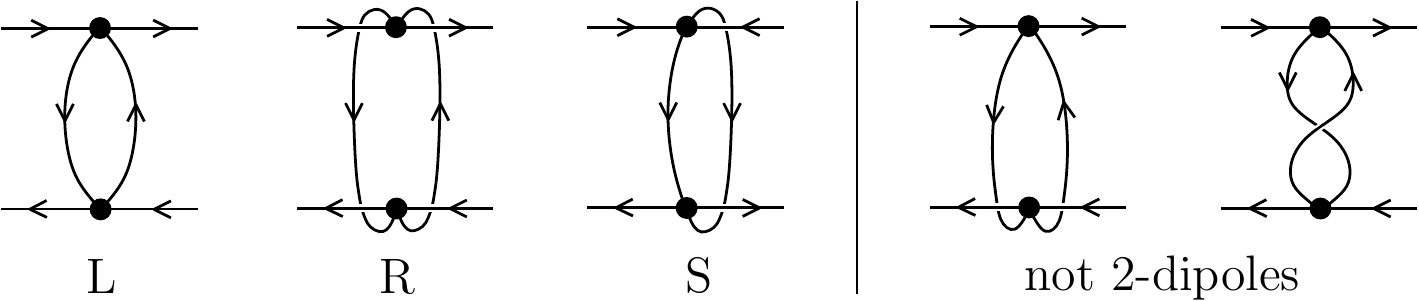}
\caption{Left: the three types of dipoles. Right: examples of double edges that do not form a~dipole.}\label{fig:dipoles}
\end{figure}

Let us also notice that a face of length two is always incident to two distinct vertices, except in the
 MO-graphs that are made of one vertex and two tadpoles~-- the inf\/inity graphs. They have no dipole but have two faces of length two.

Let us now give the following def\/initions:

\begin{Definition}
In an MO-graph ${\mathcal{G}}$, a \textit{chain} is a sequence of dipoles $d_1,\ldots,d_p$ (not passing
by the root if the graph is rooted) such that for each $1\leq i<p$, $d_i$ and $d_{i+1}$ are connected
by two edges involving two half-edges on the same side of $d_i$ (left or right) and two half-edges on the same side (left or right) of $d_{i+1}$.
\end{Definition}

For an example of a chain and of a sequence of dipoles which is not a chain, see Fig.~\ref{fig:chain_dipoles}.

\begin{figure}[t!]
\centering
\includegraphics[width=12cm]{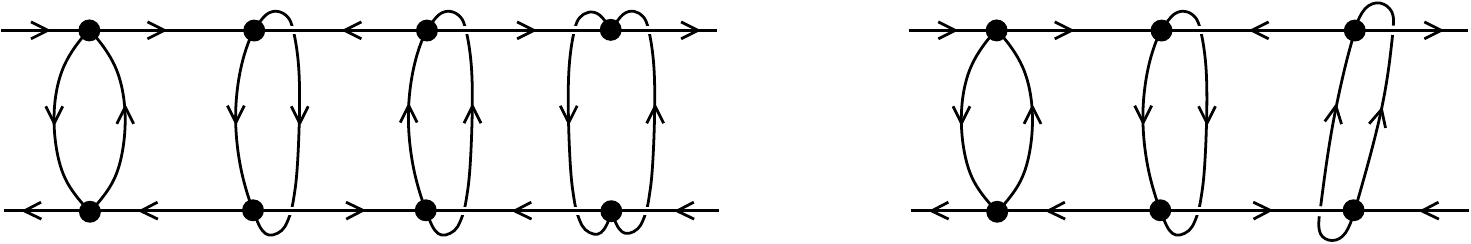}
\caption{Left: a chain of four dipoles. Right: a sequence of three dipoles that does not form a chain.}
\label{fig:chain_dipoles}
\end{figure}

Let us now give the following def\/initions:

\begin{Definition}
 An \textit{unbroken chain} is a chain for whom all the dipoles are of the same type (L, R or S).
\end{Definition}

\begin{Definition}
A {\it broken chain} is a chain which is not unbroken.
\end{Definition}

\begin{Definition}
A \textit{proper chain} is a chain of at least two dipoles.
\end{Definition}

\begin{Definition}
A \textit{maximal proper chain} is
a proper chain
for whom there is no larger chain in the given map
that the respective maximal proper chain is part of.
\end{Definition}

In a graph, one can then replace a proper chain by a {\it chain-vertex}. In Fig.~\ref{fig:chains}, we give examples of the four possible types of unbroken proper chains, replaced by corresponding chain-vertices, marked with an appropriate graphical symbol. Note that one has two distinct cases, if the number of dipoles of type S is even or odd. The bottom part of the f\/igure gives an example of a broken proper chain replaced by the corresponding chain-vertex.

\begin{figure}[t!]
\centering
\includegraphics[width=12cm]{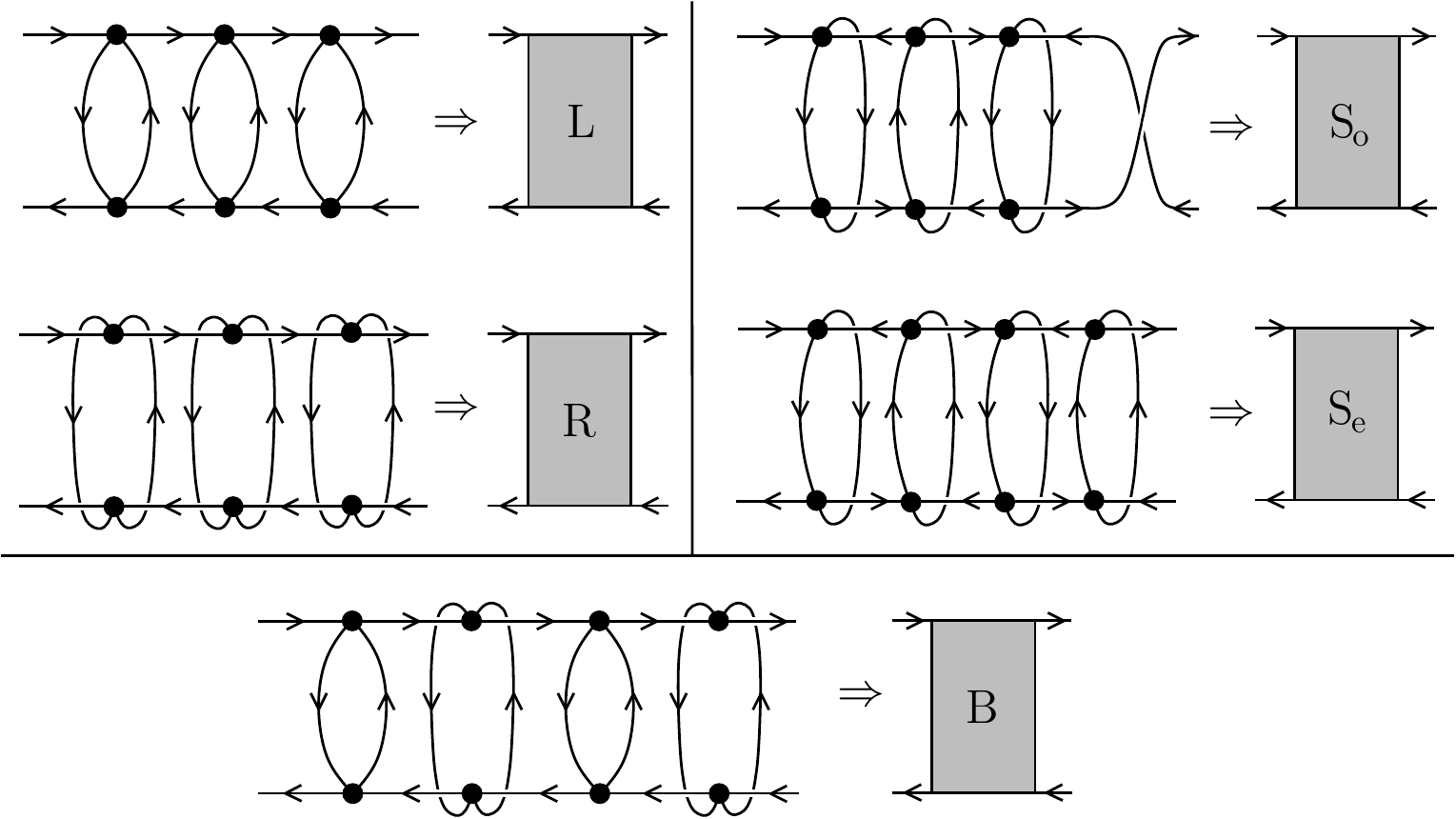}
\caption{Examples of the types of unbroken proper chains; example
of a broken proper chain.}\label{fig:chains}
\end{figure}

The corresponding conf\/igurations of strands are shown in Fig.~\ref{fig:chain_strands}. Note that these are all the possible strand conf\/igurations of the MO setting.

\begin{figure}[t!]
\centering
\includegraphics[width=12cm]{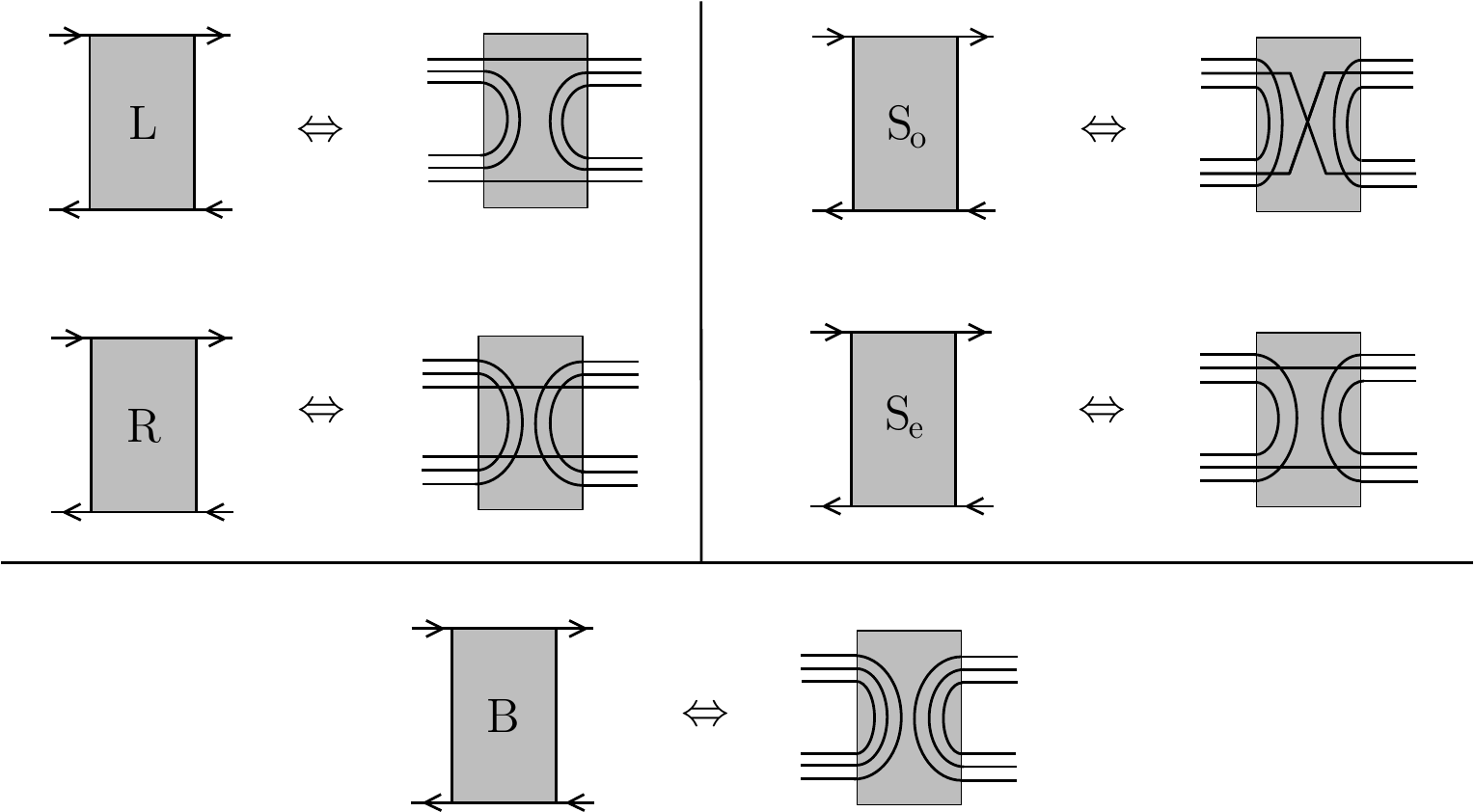}
\caption{The conf\/igurations of the strands for each type of chain-vertex.}\label{fig:chain_strands}
\end{figure}

 We can now give the following def\/inition:
\begin{Definition}
Let ${\mathcal{G}}$ be a rooted melon-free MO-graph.
The \textit{scheme} of ${\mathcal{G}}$ is the graph obtained by simultaneously replacing any maximal proper chain
of ${\mathcal{G}}$ by the corresponding {chain-vertex}.
\end{Definition}

 One can see how such a scheme is obtained from a melon-free graph from the example of
 Fig.~\ref{fig:example_scheme}.

\begin{figure}[t!]
\centering
\includegraphics[width=13cm]{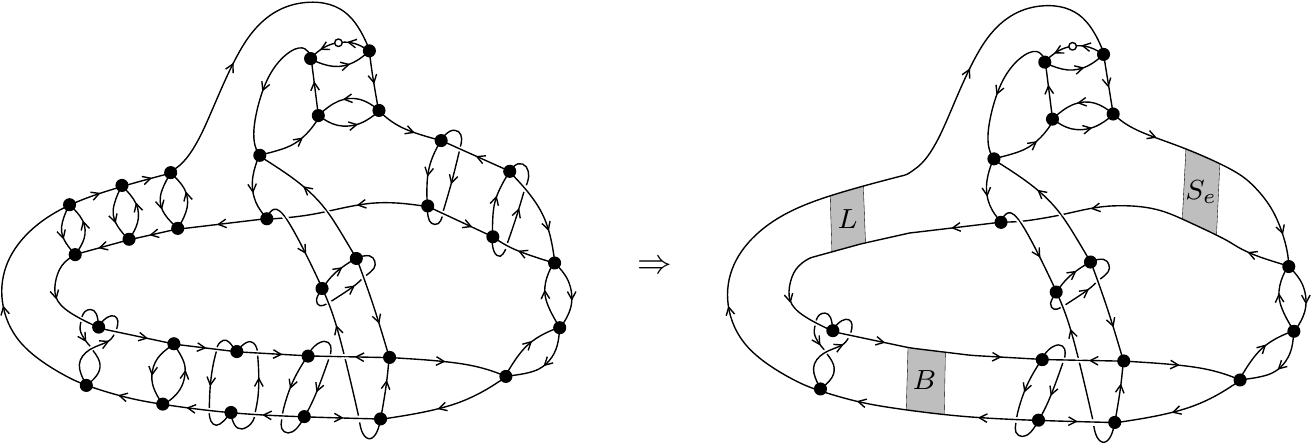}
\caption{Left: a rooted melon-free MO-graph. Right: the associated scheme.}\label{fig:example_scheme}
\end{figure}

Finally, one can give the following def\/inition:

\begin{Definition}
A \textit{reduced scheme} is a rooted melon-free MO-graph with chain-vertices
and with no proper chain.
\end{Definition}

Note that, by construction, the scheme of a rooted melon-free MO-graph (with no chain-vertices) is a reduced scheme.

The following result ensures that the degree def\/inition for MO-graphs with chain-vertices is consistent with the replacement of chains by chain-vertices:

\begin{Proposition}\label{lem:substitute}
Let ${\mathcal{G}}$ be an MO-graph with chain-vertices and let ${\mathcal{G}}'$ be an MO-graph with chain-vertices obtained from ${\mathcal{G}}$
by consistently substituting a chain-vertex
by a chain of dipoles. Then the degrees of ${\mathcal{G}}$ and ${\mathcal{G}}'$ are equal.
\end{Proposition}

One can then proove the following f\/initeness result:

\begin{Theorem}\label{prop:finite}
The set of reduced schemes of a given degree $\delta$
is finite.
\end{Theorem}

The proof (see \cite{FT} for details) relies on the following lemmas:

\begin{Lemma}\label{lem:bound1}
For each reduced scheme of degree $\delta$, the sum $N(G)$ of the numbers of dipoles
and chain-vertices satisfies the bound: $N(G)\leq 7\delta-1$.
\end{Lemma}

\begin{Lemma}\label{lem:bound2}
For $k\geq 1$ and a given degree $\delta$, there is a constant $n_{k,\delta}$
such that
 any connected unrooted MO-graph $($without chain-vertices$)$ of degree~$\delta$
with at most~$k$ dipoles has at most~$n_{k,\delta}$ vertices.
\end{Lemma}

In order to prove Theorem \ref{prop:finite} using these two lemmas, one can explicitly construct an injective map $S\mapsto G(S)$ which associates to a reduced scheme $S$ a rooted MO-graph with no chain-vertices. This construction can be done using the choice of canonical substitution of chain-vertices by proper chains given in Fig.~\ref{fig:substitute_canonical}.
This substitution
preserves the strand structure (hence the degree) and yields
an injective mapping from MO-graphs with chain-vertices to MO-graphs without chain-vertices.

\begin{figure}
\centering
\includegraphics[width=12.7cm]{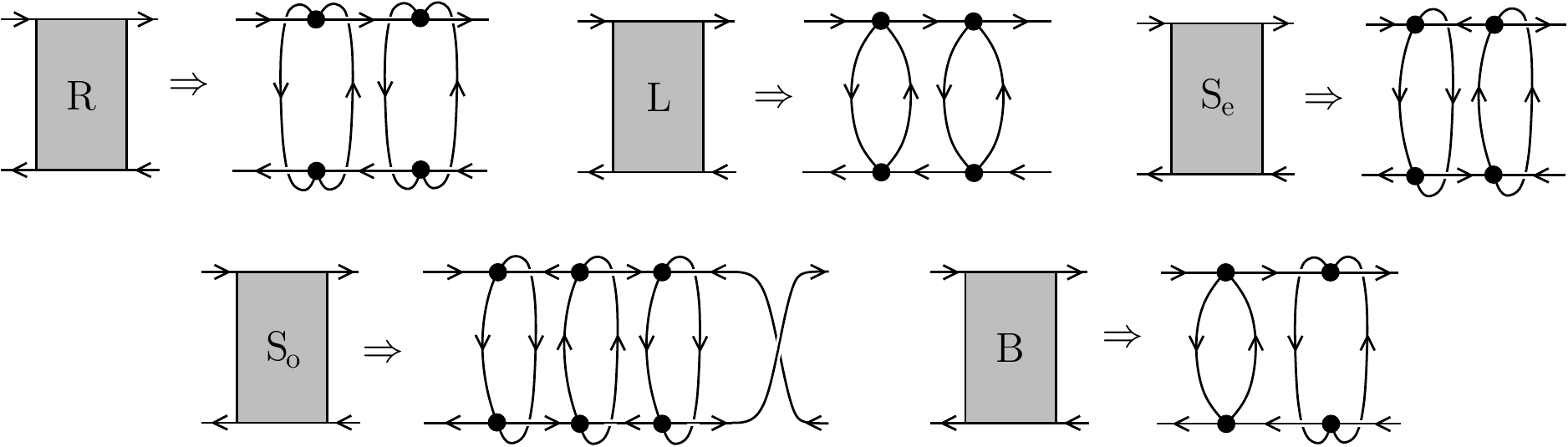}
\caption{A choice of canonical substitution leading to the map $S\mapsto G(S)$.}\label{fig:substitute_canonical}
\end{figure}

One then needs to analyze how the degree of an MO-graph with chain-vertices (not
necessarily a reduced scheme, possibly with melons) evolves when removing
\begin{enumerate}\itemsep=0pt
\item[1)] a chain-vertex,
\item[2)] a melon.
\end{enumerate}

Let ${\mathcal{G}}$ be an MO-graph with chain-vertices, and let $m$ be a chain-vertex of ${\mathcal{G}}$.
 The \emph{removal} of $m$ consists of the following operations:
\begin{enumerate}\itemsep=0pt
\item[(i)] delete $m$ from ${\mathcal{G}}$;
\item[(ii)] on each side of $m$, connect together the two
detached legs (without creating a new vertex).
\end{enumerate}
Such a removal of a chain-vertex is represented in the left part of Fig.~\ref{fig:deletion}. For the removal of a~dipole (of types L, R and S), see the right part of Fig.~\ref{fig:deletion}.

\begin{figure}\centering
\includegraphics[width=12.4cm]{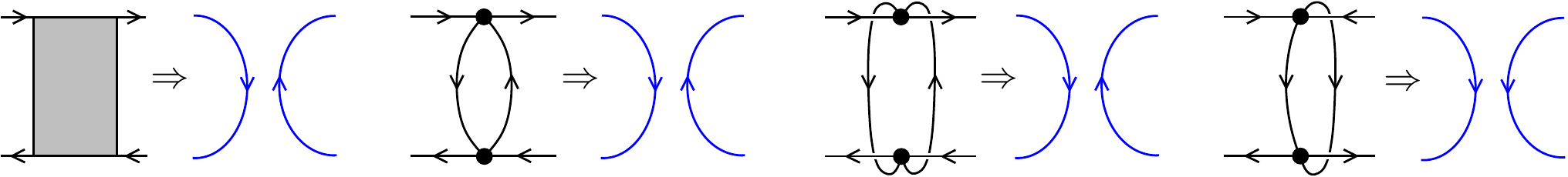}
\caption{Removal of a chain-vertex (whatever its type) and of type L, R and S.}\label{fig:deletion}
\end{figure}

The chain-vertex $m$ is said to be{\samepage
\begin{enumerate}\itemsep=0pt
\item[1)] \textit{non-separating} if ${\mathcal{G}}'$ has the same number of connected components as ${\mathcal{G}}$, and
\item[2)] \textit{separating} otherwise (in which case ${\mathcal{G}}'$ has one more connected component).
\end{enumerate}}

By carefully counting the modif\/ications made to the number of faces, vertices and connected components, one can prove the following lemma (see again \cite{FT} for details):

\begin{Lemma}\label{lem:removals}\quad
\begin{itemize}\itemsep=0pt
\item The degree is unchanged when removing a separating chain-vertex or a separating dipole $($the degree is distributed among the resulting components$)$.
\item The degree decreases by three when removing a non-separating broken chain-vertex.
\item The degree decreases by one or three when removing a non-separating unbroken chain-vertex or a non-separating dipole.
\end{itemize}
\end{Lemma}

In order to prove Lemma \ref{lem:bound1}, one needs to analyze the process of iterative removal of dipoles and chain vertices (f\/irst the non-separating ones and then the separating ones)~-- see again \cite{FT} for details. One can show that this process leads to a decorated tree of components.

An illustration of this process transforming a scheme into some decorated tree is given in Fig.~\ref{fig:from_scheme_to_tree}.
\begin{figure}\centering
\includegraphics[width=12.5cm]{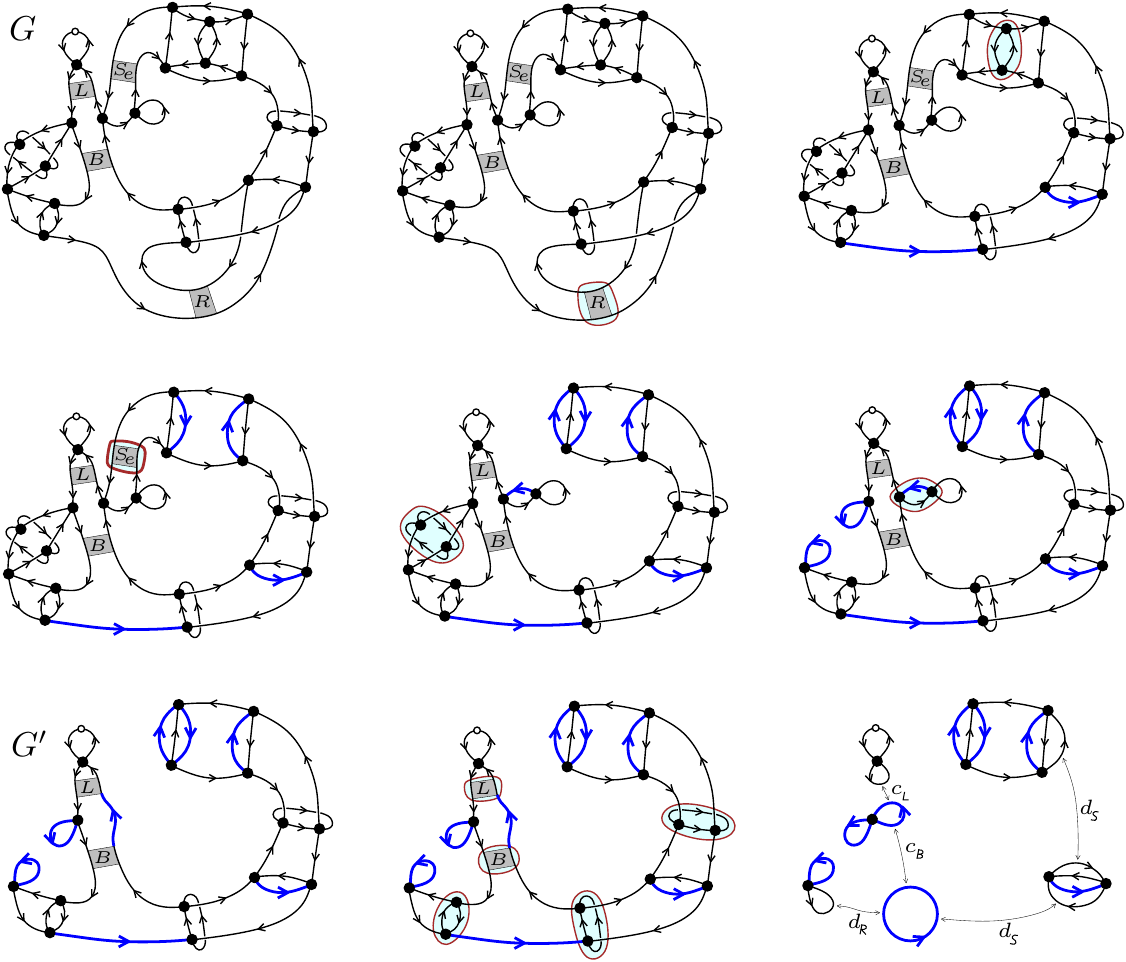}
\caption{An illustration of the process transforming a scheme into some decorated tree of components.}\label{fig:from_scheme_to_tree}
\end{figure}
The f\/irst drawing shows a reduced scheme ${\mathcal{G}}$. Then, iteratively, one removes
at each step a non-separating dipole or a non-separating chain-vertex (at each step the
non-separating dipole or chain-vertex to be removed next is surrounded). Let ${\mathcal{G}}'$ be the
MO-graph with chain-vertices thus obtained (where colored edges are drawn bolder).
As the last two drawings show,
the removal of uncolored dipoles and chain-vertices (which are all separating) of ${\mathcal{G}}'$ yields a~tree of components (a~tree edge is labelled $c_x$ if it comes from a chain-vertex of type~$x$ and is labelled~$d_x$ if it comes from an uncolored dipole of type~$x$).

\subsection{Generating functions, asymptotic enumeration and dominant schemes}

Let us start this subsection by recalling that the generating function of rooted melonic graphs writes:
\begin{gather}\label{gfmelon}
T\big(\lambda^2\big)=1+\lambda^2\big(T\big(\lambda^2\big)\big)^4
\end{gather}
(see Fig.~\ref{insertion}).

Let now $\cS_{\delta}$ be the (f\/inite) set of reduced schemes
of degree $\delta$. For each $S\in\cS_{\delta}$, let $G_S^{(\delta)}(u)$ be the generating function of
 rooted melon-free MO-graphs of reduced scheme~$S$.

Let $p$ be half the number of non-root standard vertices of $S$, $b$
the number of broken chain-vertices, $a$ the number of unbroken chain-vertices of type $L$ or $R$,
$s_{\rm e}$ the number of even straight chain-vertices, and $s_{\rm o}$ the number of odd straight chain-vertices.
The generating functions for
\begin{itemize}\itemsep=0pt
\item unbroken chains of type L (resp.~R) is $(\lambda^2)^2/(1-\lambda^2)$,
\item the one for even straight chains is $(\lambda^2)^2/(1-(\lambda^2)^2)$,
\item the one for odd straight chains is $(\lambda^2)^3/(1-(\lambda^2)^2)$,
\item and the one for broken chains is $(3\lambda^2)^2/(1-3(\lambda^2))-3(\lambda^2)^2/(1-\lambda^2)=6(\lambda^2)^2/((1-3\lambda^2)(1-\lambda^2))$.
\end{itemize}
Putting all of this together leads to
\begin{gather*}
G_S^{(\delta)}\big(\lambda^2\big)=\big(\lambda^2\big)^p\frac{(\lambda^2)^{2a}}{(1-\lambda^2)^a}\frac{(\lambda^2)^{2s_{\rm e}}}{(1-(\lambda^2)^2)^{s_{\rm e}}}\frac{(\lambda^2)^{3s_{\rm o}}}{(1-(\lambda^2)^2)^{s_{\rm o}}}\frac{6^b (\lambda^2)^{2b}}{(1-3\lambda^2)^b(1-\lambda^2)^b}.
\end{gather*}
Denoting by $c$ the total number of chain-vertices and by $s=s_{\rm e}+s_{\rm o}$ the total number of straight chain-vertices, this expression simplif\/ies to
\begin{gather*}
G_S^{(\delta)}\big(\lambda^2\big)=\frac{6^b (\lambda^2)^{p+2c+s_{\rm o}}}{(1-\lambda^2)^{c-s}(1-(\lambda^2)^2)^s(1-3\lambda^2)^b}.
\end{gather*}

Now, in order to take melons into considerations (see equation \eqref{gfmelon} above), recall that
a~rooted melon-free MO-graph with $2p$ non-root vertices has $4p+1$ edges (since
the root-edge is split into two edges, see the previous subsection) where one can insert a rooted
melonic subgraph.

Let us now def\/ine
\begin{gather*}
U\big(\lambda^2\big):=\lambda^2T\big(\lambda^2\big)^4=T\big(\lambda^2\big)-1.
\end{gather*}
 The generating function $F_S^{(\delta)}(\lambda^2)$ of rooted MO-graphs of reduced scheme $S$ is then given by
\begin{gather}\label{gfF}
F_S^{(\delta)}\big(\lambda^2\big)=T\big(\lambda^2\big)\frac{6^bU(\lambda^2)^{p+2c+s_{\rm o}}}{(1-U(\lambda^2))^{c-s}(1-U(\lambda^2)^2)^s(1-3U(\lambda^2))^b}.
\end{gather}
Finally,
the generating function $F^{(\delta)}(\lambda^2)$ of rooted MO-graphs of degree $\delta$ is simply
given by
\begin{gather*}
F^{(\delta)}\big(\lambda^2\big)=\sum_{S\in\cS_{\delta}}F_S^{(\delta)}\big(\lambda^2\big)
\end{gather*}
(see again \cite{FT} for details).

The melon generating function has its main singularity at
\begin{gather*}
\lambda_c^2=\frac{3^3}{2^8}.
\end{gather*}
Moreover, $T(\lambda^2_c)=\frac 43$, and
\begin{gather*}
\big(1-3U\big(\lambda^2\big)\big)^{-b}\sim_{\lambda\to \lambda_c}\big(2^{3/2}3^{-1/2}\big)^{-b}\big(1-\lambda^2/\lambda_c^2\big)^{-b/2}.
\end{gather*}
Therefore, using the expression \eqref{gfF}, one concludes that {\it the dominant schemes are those for which
the number of broken chains $b$ is maximized} (the larger $b$, the larger singularity order).

Using now an appropriate algorithm of iterative removal of broken chains (see again \cite{FT} for details), one obtains some tree of components. This allows to prove the following bound
\begin{gather*}
b\le 4\delta -1
\end{gather*}
on the number of broken chains.
If this bound is saturated then:
\begin{itemize}\itemsep=0pt
\item all broken chains are separating,
\item the component containing the root has degree zero,
\item all the components of positive degree and the component containing the root are leaves of the tree and the remaining components (of degree zero) have three neighbors in the tree,
\item {\it all positive degree components have degree $1/2$}.
\end{itemize}

One can then prove (see again \cite{FT} for details) the following correspondence between these dominant schemes and rooted binary trees:

\begin{Theorem}
The dominant schemes of degree $\delta$ arise from rooted binary trees with
\begin{itemize}\itemsep=0pt
\item $2\delta +1 $ leaves,
\item $2\delta - 1$ inner nodes,
\item $4\delta -1$ edges,
\end{itemize}
where
\begin{itemize}\itemsep=0pt
\item the root-leaf is occupied by the rooted cycle-graph,
\item the $2\delta$ leaves are occupied by a infinity graph,
\item the $2\delta -1$ inner nodes are occupied by the cycle graph or the quadruple edge graph,
\item the $4\delta -1$ edges are occupied by separating broken chain-vertices.
\end{itemize}
Moreover, each rooted binary tree with $2\delta +1 $ leaves yields $2^{6\delta -2}$ dominant schemes.
\end{Theorem}

In Fig.~\ref{fig:dominant}, one has an illustration of this correspondence for the case $\delta=2$.
\begin{figure}[h!]
\centering
\includegraphics[width=8cm]{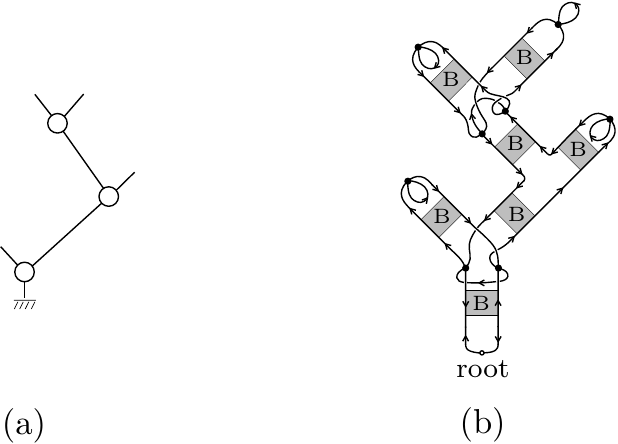}
\caption{(a) A rooted binary tree with $5$ leaves. (b) One of the $2^{10}$ dominant schemes arising from it.}\label{fig:dominant}
\end{figure}

Using Corollary VI.1 of \cite{fs}, one has:

\begin{Theorem}
\label{prop:asympt}
For $\delta$ and $n$ in $\tfrac12\mathbb{Z}_+$, let $a_n^{(\delta)}$ be the number of rooted MO-graphs with $2n$ vertices and degree $\delta$. Then, $\delta$ being fixed, for $n\in \delta+\mathbb{Z}$ $($and $\Gamma(\cdot)$ denoting the Euler gamma function$)$, one has
\begin{gather}\label{eq:asympt_degree}
a_n^{(\delta)}\sim \mathrm{Cat}_{2\delta-1}\cdot\frac{3^{\delta-3/2}}{2^{2\delta-5/2}}\cdot\frac{n^{2\delta-3/2}}{\Gamma(2\delta-1/2)}\cdot\big(2^8/3^3\big)^n\qquad \mathrm{as}\quad n\to\infty,
\end{gather}
and $a_{n+1/2}^{\delta}=O\big(a_{n}^{\delta}/\sqrt{n}\big)$ as $n\to\infty$.
\end{Theorem}

Recall that one can def\/ine a tadpole erasing mechanism which keeps the degree constant (see above). Using this mechanism for the $2\delta$ tadpoles of a dominant scheme, one can prove the following planarity result:

\begin{Theorem}
Each rooted melon-free MO-graph whose reduced scheme is dominant is planar.
\end{Theorem}

An illustration of this planarity result is shown in Fig.~\ref{fig:planar_redraw}
\begin{figure}[h!]
\centering
\includegraphics[width=8cm]{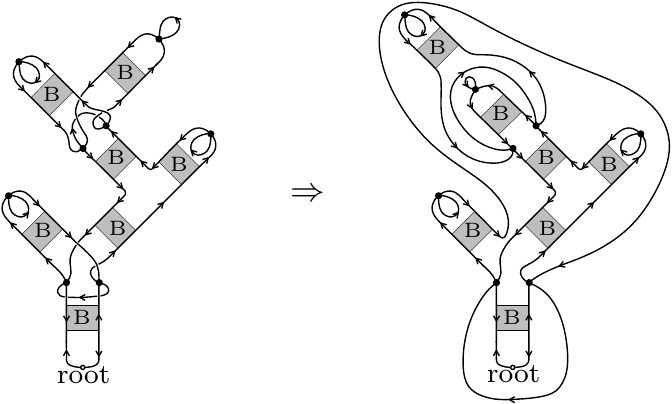}
\caption{Planar redrawing of the scheme of Fig.~\ref{fig:dominant}.}\label{fig:planar_redraw}
\end{figure}

Moreover, one has:

\begin{Corollary}
\label{coro:almost_sure}
For each fixed $\delta\in\tfrac12\mathbb{Z}_+$ and for $n\in\delta+\mathbb{Z}$, the probability that a rooted MO-graph of degree $\delta$ with $2n$ vertices is planar tends to $1$ as $n\to\infty$.
\end{Corollary}
This is a direct consequence of the fact that edge-substitution by melonic components preserves planarity, since it is these graphs which dominate the asymptotic expansion.

In the planar case, rooted MO-graphs correspond (bijectively) to rooted 4-regular maps and the straight faces of the MO-graph identify to the knot-components
of the map, and the number of straight faces is equal to
\begin{gather*}
V/2+1-\delta,
\end{gather*}
with $V$ the number of vertices and $\delta$ the degree. Let us recall that straight faces are the closed circuits given by the straight stands of our tensor graphs.

One can then prove:
\begin{Proposition}
For $n\in\tfrac12\mathbb{Z}_+$ and $k\in\mathbb{Z}_+$, let $b_{n}^{(k)}$ be the number of rooted $4$-regular planar maps with $2n$ vertices and $k$ knot-components. Then $b_{n}^{(k)}=0$ for $k>n+1$. In addition, for each fixed $\delta\in\tfrac12\mathbb{Z}_+$, and for $n\in\delta+\mathbb{Z}$, $b_{n}^{(n+1-\delta)}$ has the same asymptotic estimate as $a_n^{(\delta)}$, given by~\eqref{eq:asympt_degree}.
\end{Proposition}

\section{The double scaling limit}

Using the results of the previous section, we implement here the double scaling limit of the MO model. This follows the original article~\cite{GTY}.

The contribution to the $2r$-functions are given by the dominant schemes with~$r$ root edges (see~\cite{GTY}). Recall here that all the chains of these schemes are broken chains.

Generalizing the results of the previous section, one can show (see again \cite{GTY} for details) that the set of dominant schemes with degree $\delta$ and $r$ roots, $ {\cal S}^{\text{dom}}_{\delta, r}$, consists in binary trees with~$r$ univalent root vertices and another $2\delta$ univalent vertices. Such trees have $2\delta + r -2$ three valent internal vertices and $ 4\delta + 2r - 3$ edges. The leading singular contribution to the $2r$-point function is then
\begin{gather}
 \mathfrak{K}^{(1)\text{sing}}_{2r} = N^{3(1-r)} \sum_{\delta \in \NN / 2} N^{ - \delta }
\sum_{S \in {\cal S}^{\text{dom}}_{\delta, r} }
T\big(\lambda^2\big)^r \big[2^2 U\big(\lambda^2\big) \big]^{\delta} \big[1+3 U\big(\lambda^2\big)\big]^{2\delta+r-2} \nonumber\\
\hphantom{\mathfrak{K}^{(1)\text{sing}}_{2r} =}{} \times \left( \frac{
 6 [ U(\lambda^2)]^{ 2 } } { [1-U(\lambda^2) ] [1-3U(\lambda^2)] } \right)^{4\delta + 2r-3} ,\label{eq:domi}
\end{gather}
where for the two point function ($r=1$) one needs to add the contribution $ T(\lambda^2)$ of the degenerate dominant scheme consisting
in a unique root vertex.

In the {\it double scaling limit} one compensates the $1/N$ suppression of the higher order terms in the series in equation~\eqref{eq:domi}
by the enhancement at criticality due to the $ [1-3U(\lambda^2)]^{-1} $ factors. This is achieved by sending at the same time $N$ to inf\/inity and
$\lambda$ to criticality while keeping the double scaling parameter
\begin{gather*}
N^{\frac{1}{2}} \left( 1 - \frac{\lambda^2}{\lambda_c^2}\right) \equiv \kappa^{-1} ,
\end{gather*}
f\/ixed. In the double scaling regime we get
\begin{gather*}
T\big(\lambda^2\big) \sim \frac{4}{3} \left( 1 - \sqrt{\frac{1}{6\kappa\sqrt{N}} }\right) ,\qquad
1 - 3U\big(\lambda^2\big) \sim \sqrt{\frac{8}{3\kappa\sqrt{N}}}.
\end{gather*}

\subsection{The two-point function}

The dominant schemes are rooted binary trees which are counted by the Catalan numbers: there are ${\rm Cat}_{n-1} $ such trees with $n-1$ three valent vertices, hence with $n$ non root univalent vertices. The degree of the scheme is $ \delta = n/2$ and taking into account the contribution of
the degenerate scheme consisting in a unique root vertex we obtain
\begin{gather*}
 \mathfrak{K}^{\text{sing}}_{2} = T\big(\lambda^2\big) + \sum_{n\ge 1} {\rm Cat}_{n-1} \frac{1}{N^{ \frac{n}{2} } } T\big(\lambda^2\big) \big[2^2 U\big(\lambda^2\big) \big]^{\frac{n}{2}} \big[1+3 U\big(\lambda^2\big) \big]^{n-1} \\
\hphantom{\mathfrak{K}^{\text{sing}}_{2} =}{} \times \left( \frac{
 6 [ U(\lambda^2)]^{ 2 } } { [1-U(\lambda^2) ] [1-3U(\lambda^2)] } \right)^{2n -1 } \\
\hphantom{\mathfrak{K}^{\text{sing}}_{2} }{} = T\big(\lambda^2\big) \left( 1 + \frac{1}{N^{1/2}}\frac{12\cdot U(\lambda^2)^{5/2} }{ [1-U(\lambda^2)][1-3U(\lambda^2)] } \right. \\
\hphantom{\mathfrak{K}^{\text{sing}}_{2} =}{} \times \left.\sum_{n\geq 0} {\rm Cat}_n
\left( \frac{1}{N^{1/2}}\frac{ 72\cdot U(\lambda^2)^{9/2} [1+3U(\lambda^2)] }{ [ 1-U(\lambda^2) ]^2 [ 1-3U(\lambda)]^2} \right)^n \right) ,
\end{gather*}
which in the double scaling limit becomes
\begin{gather*}
 \mathfrak{K}^{\rm DS}_2 = \frac{4}{3}\left( 1 - \sqrt{\frac{1}{6\kappa\sqrt{N}} }\right) + \frac{4}{3} \frac{1}{N^{\frac{1}{4}} } \sqrt{ \frac{\kappa}{2} } \sum_{n\geq 0} {\rm Cat}_n \left( \kappa \frac{\sqrt{3}}{2}\right)^n
 = \frac{4}{3} \left( 1 - \frac{ \sqrt{1-2\sqrt{3} \kappa} }{ N^{\frac{1}{4} } \sqrt{6 \kappa} }\right),
\end{gather*}
where the sum over $n$ converges for $\kappa<2^{-1} 3^{-\frac{1}{2}}$.

\subsection{The four-point function}

The dominant schemes are binary trees with two roots. There are again ${\rm Cat}_{n-1} $ such trees with $n-1$ three valent vertices,
but this time they have $n-1$ non root univalent vertices and degree $\delta = \frac{n-1}{2}$.
We obtain
\begin{gather*}
\mathfrak{K}^{(1)\text{sing}}_{4} = N^{-3} \sum_{n\ge 1} {\rm Cat}_{n-1} \frac{1}{N^{ \frac{n-1}{2} } } T\big(\lambda^2\big)^2 \big[2^2 U\big(\lambda^2\big) \big]^{\frac{n-1}{2}} \big[1+3 U\big(\lambda^2\big) \big]^{n-1} \\
\hphantom{\mathfrak{K}^{(1)\text{sing}}_{4} =}{}
\times \left( \frac{
 6 [ U(\lambda^2)]^{ 2 } } { [1-U(\lambda^2) ] [1-3U(\lambda^2)] } \right)^{2(n-1) + 1} \\
\hphantom{\mathfrak{K}^{(1)\text{sing}}_{4}}{} = N^{-3} T\big(\lambda^2\big)^2 \left( \frac{
 6 [ U(\lambda^2)]^{ 2 } } { [1-U(\lambda^2) ] [1-3U(\lambda^2)] } \right) \\
\hphantom{\mathfrak{K}^{(1)\text{sing}}_{4} =}{} \times \sum_{n\geq 0} {\rm Cat}_n
\left( \frac{1}{N^{1/2}}\frac{ 72\cdot U(\lambda^2)^{9/2} [1+3U(\lambda^2)] }{ [ 1-U(\lambda^2) ]^2 [ 1-3U(\lambda)]^2} \right)^n .
\end{gather*}

In the double scaling limit this becomes
\begin{gather*}
 \mathfrak{K}^{(1){\rm DS}}_{4} = N^{-3 + \frac{1}{4} } \sqrt{\kappa} \frac{8 \sqrt{3}}{ 9\sqrt{2} } \left( \frac{1 - \sqrt{1-2\sqrt{3}\kappa} }{ \sqrt{3} \kappa } \right) .
\end{gather*}

Note that $\mathfrak{K}^{(1){\rm DS}}_{4}$ is \emph{enhanced} by a factor $N^{\frac{1}{4}}$ with respect to the natural $N$ scaling
of $\mathfrak{K}^{(1)}_{4} $. This is a consequence of the fact that the singularity of
the generating function of dominant schemes boosts this four point function in double scaling.

\subsection[The $2r$-point function]{The $\boldsymbol{2r}$-point function}

The dominant schemes are binary trees with $r$ roots and the double scaling limit of $\mathfrak{K}^{(1)}_{2r}$ is
\begin{gather*}
 \mathfrak{K}^{(1){\rm DS}}_{2r} \sim N^{3(1-r)} N^{\frac{1}{4}(2r-3)} f_{2r}(\kappa),
\end{gather*}
for some function $f_{2r}$ depending only on the double scaling parameter~$\kappa$. As it was the case for the four-point function, the higher point functions
are also boosted in double scaling with respect to their natural scaling in~$N$.

Let us emphasize on the fact that the functions $f_{2r}(\kappa) $ are convergent for
$2\sqrt{3}\kappa < 1 $ and exhibit a square root singularity at the critical double scaling coupling $ \kappa_c= 2^{-1} 3^{-\frac{1}{2}}$.

\section{Concluding remarks and perspectives}

We have reviewed in this paper the def\/inition of the MO random tensor model and several QFT results (such as the large $N$ expansion and the double scaling limit) related to this. From a~combinatorial point of view, the dominant schemes at any order of the $1/N$ expansion have been identif\/ied and carefully studied (their shapes being shown to be naturally associated to rooted binary trees).

A f\/irst perspective for future work appears to us to be the study from a probabilistic point of view of these MO dominant schemes. More concretely, it would be interesting to investigate to what phases these recently discovered conf\/igurations correspond. Let us recall that melon graphs correspond, from this point of view, to branched polymers \cite{GR-AIHP}.

A second, and maybe immediate, perspective for future work is the study of renormalizability of this type of tensor models, in the spirit of \cite{BGR}.
Obviously, one can also use non-trivial holonomy group data, thus making the contact with the original GFT framework, in the spirit of~\cite{teza-Sylvain} or~\cite{Sylvain-AIHPD}. The drawback would of course be that one would need to add an {\it ad-hoc} propagator in the quadratic part of the MO action.

A distinct line of research is then given by the extension of the MO framework to classes of tensor models taking into consideration larger and larger classes of Feynman graphs. As already emphasized here, with respect to the celebrated colored-like tensor models, the MO model enlarges the class of Feynman graphs to be considered, but there is no reason for whom this model would be the one leading to the largest possible class of Feynman graphs.

On a more general basis, leaving aside the particular class of tensor models one chooses to work with, it appears to us that tensor models are now at a crucial point of their development. Thus, a f\/irst set of QFT questions have been answered (the implementation of the large $N$ and double scaling limit mechanisms) for both the colored-like and MO frameworks.

Nevertheless, crucial questions remain open at this point. From the physics perspective, maybe the most important ones are the one of the continuous limit of random tensor models and of the relation of this models to quantum gravity. Ideally, this would be obtained in the same way it is obtained in the 2D case, where the double-scaling mechanism allows to access the regime of f\/inite, non-vanishing Newton constant (since the large $N$ limit corresponds to the vanishing limit of the Newton constant, while the limit $\lambda\to\lambda_c$ corresponds to the large volume limit, see again the review~\cite{DiFrancesco:1993nw} or the papers \cite{dm1, dm2, dm3}).

From the combinatorics and mathematics perspective, an important result would be a ge\-ne\-ra\-lization of Schaef\/fer~\cite{sch} or BDFG bijections~\cite{bdfg}, generalization which should allow to access information on the geodesic. Let us recall that, at~2D, the particular interest of these bijections it is exactly the fact that they allow to obtain geodesic information, information which is not accessible through standard QFT techniques.

Another perspective is on our opinion the possibility of using tensor integrals to obtain coun\-ting theorems for maps in three (or higher) dimensions. This would be of particular importance for the f\/ield of combinatorics, where no such counting results are not known through standard combinatorial techniques.

Nevertheless, all these targets appear to us as particularly dif\/f\/icult to attend, since quantum gravity, geometry or topology in dimension higher than two are very much involved, both from a~conceptual and from a technical point of view. Several attempts of mastering these concepts have been made so far, using various angles of attacks, in mathematics or physics. It thus remains for tensor models to obtain progress in these important directions of research in the future.

\subsection*{Acknowledgements}
The author is partially supported by the grants ANR JCJC CombPhysMat2Tens and PN 09 37 01 02.

\newpage

\pdfbookmark[1]{References}{ref}
\LastPageEnding


\begin{thebibliography}{99}
\footnotesize\itemsep=0pt

\bibitem{t1}
Ambj{\o}rn J., Durhuus B., J{\'o}nsson T., Three-dimensional simplicial quantum
 gravity and generalized matrix models, \href{http://dx.doi.org/10.1142/S0217732391001184}{\textit{Modern Phys. Lett.~A}}
 \textbf{6} (1991), 1133--1146.

\bibitem{ART}
Avohou R.C., Rivasseau V., Tanasa A., Renormalization and {H}opf algebraic
 structure of the f\/ive-dimensional quartic tensor f\/ield theory,
 \href{http://dx.doi.org/10.1088/1751-8113/48/48/485204}{\textit{J.~Phys.~A: Math. Theor.}} \textbf{48} (2015), 485204, 20~pages,
 \href{http://arxiv.org/abs/1507.03548}{arXiv:1507.03548}.

\bibitem{BaratinOriti}
Baratin A., Oriti D., Ten questions on Group Field Theory (and their tentative
 answers, \href{http://dx.doi.org/10.1088/1742-6596/360/1/012002}{\textit{J.~Phys. Conf. Ser.}} \textbf{360} (2012), 012002, 10~pages,
 \href{http://arxiv.org/abs/1112.3270}{arXiv:1112.3270}.

\bibitem{BGR-AIHPD}
Ben~Geloun J., Ramgoolam S., Counting tensor model observables and branched
 covers of the 2-sphere, \href{http://dx.doi.org/10.4171/AIHPD/4}{\textit{Ann. Inst. Henri Poincar\'e~D}} \textbf{1}
 (2014), 77--138, \href{http://arxiv.org/abs/1307.6490}{arXiv:1307.6490}.

\bibitem{BGR}
Ben~Geloun J., Rivasseau V., A renormalizable 4-dimensional tensor f\/ield
 theory, \href{http://dx.doi.org/10.1007/s00220-012-1549-1}{\textit{Comm. Math. Phys.}} \textbf{318} (2013), 69--109,
 \href{http://arxiv.org/abs/1111.4997}{arXiv:1111.4997}.

\bibitem{BC1-AIHPD2}
Bonzom V., Combes F., The calculation of expectation values in {G}aussian
 random tensor theory via meanders, \href{http://dx.doi.org/10.4171/AIHPD/13}{\textit{Ann. Inst. Henri Poincar\'e~D}}
 \textbf{1} (2014), 443--485, \href{http://arxiv.org/abs/1310.3606}{arXiv:1310.3606}.

\bibitem{BC2-AIHPD2}
Bonzom V., Combes F., Tensor models from the viewpoint of matrix models: the
 cases of loop mo\-dels on random surfaces and of the {G}aussian distribution,
 \href{http://dx.doi.org/10.4171/AIHPD/14}{\textit{Ann. Inst. Henri Poincar\'e~D}} \textbf{2} (2015), 1--47,
 \href{http://arxiv.org/abs/1304.4152}{arXiv:1304.4152}.

\bibitem{bdfg}
Bouttier J., Di~Francesco P., Guitter E., Geodesic distance in planar graphs,
 \href{http://dx.doi.org/10.1016/S0550-3213(03)00355-9}{\textit{Nuclear Phys.~B}} \textbf{663} (2003), 535--567,
 \href{http://arxiv.org/abs/cond-mat/0303272}{cond-mat/0303272}.

\bibitem{dm1}
Br{\'e}zin {\'E}., Kazakov V.A., Exactly solvable f\/ield theories of closed
 strings, \href{http://dx.doi.org/10.1016/0370-2693(90)90818-Q}{\textit{Phys. Lett.~B}} \textbf{236} (1990), 144--150.

\bibitem{teza-Sylvain}
Carrozza S., Tensorial methods and renormalization in group f\/ield theories,
 \href{http://dx.doi.org/10.1007/978-3-319-05867-2}{Springer Theses}, Springer, Cham, 2014, \href{http://arxiv.org/abs/1310.3736}{arXiv:1310.3736}.

\bibitem{Sylvain-AIHPD}
Carrozza S., Discrete renormalization group for {$\rm SU(2)$} tensorial group
 f\/ield theory, \href{http://dx.doi.org/10.4171/AIHPD/15}{\textit{Ann. Inst. Henri Poincar\'e~D}} \textbf{2} (2015),
 49--112, \href{http://arxiv.org/abs/1407.4615}{arXiv:1407.4615}.

\bibitem{DGR}
Dartois S., Gurau R., Rivasseau V., Double scaling in tensor models with a
 quartic interaction, \href{http://dx.doi.org/10.1007/JHEP09(2013)088}{\textit{J.~High Energy Phys.}} \textbf{2013} (2013),
 no.~9, 088, 33~pages, \href{http://arxiv.org/abs/1307.5281}{arXiv:1307.5281}.

\bibitem{DRT}
Dartois S., Rivasseau V., Tanasa A., The {$1/N$} expansion of multi-orientable
 random tensor models, \href{http://dx.doi.org/10.1007/s00023-013-0262-8}{\textit{Ann. Henri Poincar\'e}} \textbf{15} (2014),
 965--984, \href{http://arxiv.org/abs/1301.1535}{arXiv:1301.1535}.

\bibitem{DiFrancesco:1993nw}
Di~Francesco P., Ginsparg P., Zinn-Justin J., {$2$}{D} gravity and random
 matrices, \href{http://dx.doi.org/10.1016/0370-1573(94)00084-G}{\textit{Phys. Rep.}} \textbf{254} (1995), 1--133,
 \href{http://arxiv.org/abs/hep-th/9306153}{hep-th/9306153}.

\bibitem{dm2}
Douglas M.R., Shenker S.H., Strings in less than one dimension, \href{http://dx.doi.org/10.1016/0550-3213(90)90522-F}{\textit{Nuclear
 Phys.~B}} \textbf{335} (1990), 635--654.

\bibitem{fs}
Flajolet P., Sedgewick R., Analytic combinatorics, \href{http://dx.doi.org/10.1017/CBO9780511801655}{Cambridge University Press},
 Cambridge, 2009.

\bibitem{freidel}
Freidel L., Group f\/ield theory: an overview, \href{http://dx.doi.org/10.1007/s10773-005-8894-1}{\textit{Internat.~J. Theoret.
 Phys.}} \textbf{44} (2005), 1769--1783, \mbox{\href{http://arxiv.org/abs/hep-th/0505016}{hep-th/0505016}}.

\bibitem{fgo}
Freidel L., Gurau R., Group f\/ield theory renormalization in the 3D case: power
 counting of divergences, \href{http://dx.doi.org/10.1103/PhysRevD.80.044007}{\textit{Phys. Rev.~D}} \textbf{80} (2009), 044007,
 20~pages, \href{http://arxiv.org/abs/0905.3772}{arXiv:0905.3772}.

\bibitem{FT}
Fusy E., Tanasa A., Asymptotic expansion of the multi-orientable random tensor
 model, \textit{Electron.~J. Combin.} \textbf{22} (2015), 1.52, 30~pages,
 \href{http://arxiv.org/abs/1408.5725}{arXiv:1408.5725}.

\bibitem{dm3}
Gross D.J., Migdal A.A., Nonperturbative two-dimensional quantum gravity,
 \href{http://dx.doi.org/10.1103/PhysRevLett.64.127}{\textit{Phys. Rev. Lett.}} \textbf{64} (1990), 127--130.

\bibitem{largeN}
Gurau R., The {$1/N$} expansion of colored tensor models, \href{http://dx.doi.org/10.1007/s00023-011-0101-8}{\textit{Ann. Henri
 Poincar\'e}} \textbf{12} (2011), 829--847, \href{http://arxiv.org/abs/1011.2726}{arXiv:1011.2726}.

\bibitem{color}
Gurau R., Colored group f\/ield theory, \href{http://dx.doi.org/10.1007/s00220-011-1226-9}{\textit{Comm. Math. Phys.}} \textbf{304}
 (2011), 69--93, \href{http://arxiv.org/abs/0907.2582}{arXiv:0907.2582}.

\bibitem{GR}
Gurau R., Rivasseau V., The $1/N$ expansion of colored tensor models in
 arbitrary dimension, \href{http://dx.doi.org/10.1209/0295-5075/95/50004}{\textit{Europhys. Lett.}} \textbf{95} (2011), 50004,
 5~pages, \href{http://arxiv.org/abs/1101.4182}{arXiv:1101.4182}.

\bibitem{GR-Sigma}
Gurau R., Ryan J.P., Colored tensor models~-- a review, \href{http://dx.doi.org/10.3842/SIGMA.2012.020}{\textit{SIGMA}}
 \textbf{8} (2012), 020, 78~pages, \href{http://arxiv.org/abs/1109.4812}{arXiv:1109.4812}.

\bibitem{GR-AIHP}
Gurau R., Ryan J.P., Melons are branched polymers, \href{http://dx.doi.org/10.1007/s00023-013-0291-3}{\textit{Ann. Henri
 Poincar\'e}} \textbf{15} (2014), 2085--2131, \href{http://arxiv.org/abs/1302.4386}{arXiv:1302.4386}.

\bibitem{GS}
Gurau R., Schaef\/fer G., Regular colored graphs of positive degree,
 \href{http://arxiv.org/abs/1307.5279}{arXiv:1307.5279}.

\bibitem{GTY}
Gurau R., Tanasa A., Youmans D.R., The double scaling limit of the
 multi-orientable tensor model, \href{http://dx.doi.org/10.1209/0295-5075/111/21002}{\textit{Europhys. Lett.}} \textbf{111} (2015),
 21002, 6~pages, \href{http://arxiv.org/abs/1505.00586}{arXiv:1505.00586}.

\bibitem{oriti}
Oriti D. (Editor), Approaches to quantum gravity: toward a new understanding of
 space, time and matter, Cambridge University Press, Cambridge, 2009.

\bibitem{oriti2012}
Oriti D., The quantum geometry of tensorial group f\/ield theories, in Symmetries
 and Groups in Contemporary Physics, \href{http://dx.doi.org/10.1142/9789814518550_0051}{\textit{Nankai Ser. Pure Appl. Math.
 Theoret. Phys.}}, Vol.~11, World Sci. Publ., Hackensack, NJ, 2013, 379--384,
 \href{http://arxiv.org/abs/1211.5714}{arXiv:1211.5714}.

\bibitem{RT-CK}
Raasakka M., Tanasa A., Combinatorial {H}opf algebra for the {B}en
 {G}eloun--{R}ivasseau tensor f\/ield theory, \textit{S\'em. Lothar. Combin.}
 \textbf{70} (2013), B70d, 29~pages, \href{http://arxiv.org/abs/1306.1022}{arXiv:1306.1022}.

\bibitem{RT}
Raasakka M., Tanasa A., Next-to-leading order in the large {$N$} expansion of
 the multi-orientable random tensor model, \href{http://dx.doi.org/10.1007/s00023-014-0336-2}{\textit{Ann. Henri Poincar\'e}}
 \textbf{16} (2015), 1267--1281, \href{http://arxiv.org/abs/1310.3132}{arXiv:1310.3132}.

\bibitem{rev-riv-ncqft}
Rivasseau V., Non-commutative renormalization, in Quantum Spaces, \href{http://dx.doi.org/10.1007/978-3-7643-8522-4_2}{\textit{Prog.
 Math. Phys.}}, Vol.~53, Birkh\"auser, Basel, 2007, 19--107.

\bibitem{rev-riv}
Rivasseau V., The tensor track,~{III}, \href{http://dx.doi.org/10.1002/prop.201300032}{\textit{Fortschr. Phys.}} \textbf{62}
 (2014), 81--107, \href{http://arxiv.org/abs/1311.1461}{arXiv:1311.1461}.

\bibitem{t2}
Sasakura N., Tensor model for gravity and orientability of manifold,
 \href{http://dx.doi.org/10.1142/S0217732391003055}{\textit{Modern Phys. Lett.~A}} \textbf{6} (1991), 2613--2623.

\bibitem{sch}
Schaef\/fer G., Bijective census and random generation of {E}ulerian planar maps
 with prescribed vertex degrees, \textit{Electron.~J. Combin.} \textbf{4}
 (1997), 20, 14~pages.

\bibitem{praa}
Tanasa A., Combinatorics of random tensor models, \textit{Proc. Rom. Acad.
 Ser.~A Math. Phys. Tech. Sci. Inf. Sci.} \textbf{13} (2012), 27--31,
 \href{http://arxiv.org/abs/1203.5304}{arXiv:1203.5304}.

\bibitem{original}
Tanasa A., Multi-orientable group f\/ield theory, \href{http://dx.doi.org/10.1088/1751-8113/45/16/165401}{\textit{J.~Phys.~A: Math.
 Theor.}} \textbf{45} (2012), 165401, 19~pages, \href{http://arxiv.org/abs/1109.0694}{arXiv:1109.0694}.

\bibitem{praa-mo}
Tanasa A., Tensor models, a quantum f\/ield theoretical particularization,
 \textit{Proc. Rom. Acad. Ser.~A Math. Phys. Tech. Sci. Inf. Sci.} \textbf{13}
 (2012), 225--234, \href{http://arxiv.org/abs/1211.4444}{arXiv:1211.4444}.

\end{thebibliography}
\end{document}